\documentclass[10.5pt]{article}

\usepackage{amsfonts,amssymb,amsmath,amsthm,bm}
\usepackage{latexsym,url,color,enumerate}
\usepackage{graphicx}
\usepackage{lipsum}

\newcommand\blfootnote[1]{%
  \begingroup
  \renewcommand\thefootnote{}\footnote{#1}%
  \addtocounter{footnote}{-1}%
  \endgroup
}

\usepackage{amsfonts,amsmath,amssymb,amsthm,caption,algorithm,algorithmic}
\usepackage{lipsum}%footnote
\usepackage{amsbsy}
\usepackage{mathrsfs}
\usepackage{graphicx,multirow,bm}
\usepackage{color}
\usepackage{enumerate}
\usepackage{diagbox}
\usepackage{amssymb}
\usepackage{psfrag}
\usepackage{bbm}
\usepackage[usenames,dvipsnames]{pstricks}
\usepackage{epsfig}
\usepackage{pst-grad} % For gradients
\usepackage{pst-plot} % For axes
\usepackage{upgreek}

\theoremstyle{mythm}
\newtheorem{theorem}{Theorem}%[section]
\newtheorem{example}{Example}
\newtheorem{remark}{Remark}

\newtheorem{prop}{Proposition}%[section]
\newtheorem{coro}{Corollary}%[section]
\newtheorem{lemma}{Lemma}%[section]

\allowdisplaybreaks[4]

\begin{document}
\title{On the Non-Adaptive Zero-Error Capacity of the Discrete Memoryless Two-Way Channel}

\author{
Yujie Gu 
\thanks{
Y. Gu is with the Faculty of Information Science and Electrical Engineering, Kyushu University, Fukuoka, Japan \{email: gu@inf.kyushu-u.ac.jp\}. This work was done while the author was with the
Department of EE -- Systems, Tel Aviv University, Tel Aviv, Israel.} 
and Ofer Shayevitz
\thanks{
O. Shayevitz is with the
Department of EE -- Systems, Tel Aviv University, Tel Aviv, Israel \{email: ofersha@eng.tau.ac.il\}. } 

\blfootnote{This work was supported by an ERC grant no. 639573, ISF grant no. 1495/18, and JSPS grant no. 21K13830, and presented in part at the 2019 IEEE International Symposium on Information Theory~\cite{GS2019}.}
}

\date{}
\maketitle

\begin{abstract}
We study the problem of communicating over a discrete memoryless two-way channel using non-adaptive schemes, under a zero probability of error criterion. We derive single-letter inner and outer bounds for the zero-error capacity region, based on random coding, linear programming, linear codes, and the asymptotic spectrum of graphs. Among others, we provide a single-letter outer bound based on a combination of Shannon's vanishing-error capacity region and a two-way analogue of the linear programming bound for point-to-point channels, which in contrast to the one-way case, is generally better than both. Moreover, we establish an outer bound for the zero-error capacity region of a two-way channel via the asymptotic spectrum of graphs and show that this bound could be achieved for certain cases.
\end{abstract}

%\begin{IEEEkeywords}
%    Zero-error capacity, discrete memoryless two-way channel, Shannon capacity of a graph
%\end{IEEEkeywords}

\section{Introduction}

The problem of reliable communication over a {\em discrete memoryless two-way channel (DM-TWC)} was originally introduced and investigated by Shannon \cite{Shannon1961}, in a seminal paper that has marked the inception of multi-user information theory. A DM-TWC is characterized by a quadruple of finite input and output alphabets $\mathcal{X}_1$, $\mathcal{X}_2$, $\mathcal{Y}_1$, $\mathcal{Y}_2$, and
a conditional probability distribution $P_{Y_1,Y_2|X_1,X_2}(y_1,y_2|x_1,x_2)$, where
$x_1\in \mathcal{X}_1$, $x_2\in \mathcal{X}_2$, $y_1\in \mathcal{Y}_1$, $y_2\in \mathcal{Y}_2$. The channel is memoryless in the sense that
channel uses are independent, i.e., for any $i$,
\begin{align*}\label{channel-memoryless}
    &P_{Y_{1i},Y_{2i}|X^i_1,X^i_2,Y_1^{i-1},Y_2^{i-1}}(y_{1i},y_{2i}|x^i_1,x^i_2,y_1^{i-1},y_2^{i-1})= P_{Y_1,Y_2|X_1,X_2}(y_{1i},y_{2i}|x_{1i},x_{2i}).
\end{align*}

In \cite{Shannon1961}, Shannon provided inner and outer bounds for the vanishing-error capacity region of the DM-TWC, in the general setting where the users are allowed to adapt their transmissions on the fly based on past observations. We note that Shannon's inner bound is tight for non-adaptive schemes, namely when the users map their messages to codewords in advance. The non-adaptive DM-TWC is also called the \textit{restricted} DM-TWC in \cite{Shannon1956}.
Shannon's inner and outer bounds have later been improved by utilizing auxiliary random variables techniques~\cite{Han1984}, \cite{HW1989}, \cite{Zhang1986},
and sufficient conditions under which his bounds coincide have been obtained~\cite{Weng2018,Weng2019}. However, despite much effort, the capacity region of a general DM-TWC under the vanishing-error criterion remains elusive. In fact, a strong indicator for the inherent difficulty of the problem can be observed in Blackwell's binary multiplying channel, a simple, deterministic, common-output channel whose capacity remains unknown hitherto \cite{HW1989}, \cite{SP2018}, \cite{Scha1982}, \cite{Scha1983}, \cite{Zhang1986}.

In yet another seminal work, Shannon proposed and studied the zero-error capacity of the point-to-point discrete memoryless channel~\cite{Shannon1956}, also known as the Shannon capacity of a graph. This problem has been extensively studied by others, most notably in~\cite{Haemers1979} and~\cite{Lovasz1979}, yet remains generally unsolved. 
In this paper, we consider the problem of zero-error communication over a DM-TWC. We limit our discussion to the case of non-adaptive schemes, for which the capacity region is known in the vanishing-error case~\cite{Shannon1961}. Despite the obvious difficulty of the problem (the point-to-point zero-error capacity is a special case), its two-way nature adds a new combinatorial dimension that renders it interesting to study. To the best of our knowledge, this problem has not been studied before, except in  the special case of the binary multiplying channel, where upper and lower bounds on non-adaptive zero-error sum capacity have been obtained~\cite{HK1995}, \cite{J2018}, \cite{Tolhuizen}. Our bounds are partially based on generalizations of these ideas.

The problem of non-adaptive communication over a DM-TWC
can be formulated as follows. Alice and Bob would like to simultaneously convey messages $m_1\in [2^{nR_1}]$ and $m_2\in [2^{nR_2}]$ respectively to each other, over $n$ uses of the DM-TWC $P_{Y_1,Y_2|X_1,X_2}$. To that end, Alice maps her message to an input sequence (codeword)  $x^n_1\in \mathcal{X}^n_1$ using an encoding function $f_1: [2^{nR_1}]\to \mathcal{X}^n_1$, and Bob maps his message into an input sequence (codeword) $x^n_2\in \mathcal{X}^n_2$ using an encoding function $f_2: [2^{nR_2}]\to \mathcal{X}^n_2$. We call the pair of codeword collections $( f_1([2^{nR_1}]), f_2([2^{nR_2}]) )$ a {\em codebook pair}. Note that the encoding functions depend only on the messages, and not on the observed outputs during the transmission, hence the name \textit{non-adaptive}. When transmissions end, Alice and Bob observe the resulting (random) channel outputs $Y^n_1\in \mathcal{Y}^n_1$ and $Y^n_2\in \mathcal{Y}^n_2$ respectively, and attempt to decode the message sent by their counterpart, without error. When this is possible, i.e., when there exist decoding functions $\phi_1: [2^{nR_1}]\times \mathcal{Y}^n_1\to [2^{nR_2}]$ and $\phi_2: [2^{nR_2}]\times \mathcal{Y}^n_2\to [2^{nR_1}]$ such that $m_2 = \phi_1(m_1,Y_1^n)$ and $m_1 =\phi_2(m_2,Y_2^n)$, for all $m_1,m_2$, with probability one, then the codebook pair (or the encoding functions) is called {\em $(n,R_1,R_2)$ uniquely decodable}. A rate pair $(R_1,R_2)$ is \textit{achievable} for the DM-TWC if an $(n,R_1,R_2)$ uniquely decodable code exists for some $n$. The \textit{non-adaptive zero-error capacity region} of a DM-TWC $P_{Y_1,Y_2|X_1,X_2}$ is the closure of the set of all achievable rate pairs, and is denoted here by $\mathscr{C}_{ze}(P_{Y_1,Y_2|X_1,X_2})$.
Moreover, the \textit{non-adaptive zero-error sum-capacity} of a DM-TWC $P_{Y_1,Y_2|X_1,X_2}$, denoted by $\mathscr{C}^{sum}_{ze}(P_{Y_1,Y_2|X_1,X_2})$, is the supremum of sum-rate $R_1+R_2$ over all achievable rate pairs.

The main objective of this paper is to provide several single-letter outer and inner bounds on the non-adaptive zero-error capacity region of the DM-TWC.
The remainder of this paper is organized as follows.
In Section \ref{Sec-2}, we consider the confusion graphs of a DM-TWC, and discuss their performances with respect to the graph homomorphisms and the one-shot zero-error communication.
Section \ref{sec-outer-bd} is devoted to three general outer bounds of the zero-error capacity region of DM-TWC, which are based on
Shannon's vanishing-error non-adaptive capacity region, a two-way analogue of the linear programming bound for point-to-point channels, and 
the Shannon capacity of a graph. 
In Section \ref{sec-inner-bd}, we provide two general inner bounds using random coding and random linear codes respectively.
In Section \ref{sec-certain-TWC}, we
establish outer bounds for certain types of DM-TWC via 
the asymptotic spectra of graphs, and also explicitly construct the uniquely decodable codebook pairs achieving the outer bound.
Some concluding remarks are given in Section \ref{sec-conclusion}.

\section{Preliminaries}
\label{Sec-2}
%{ ***************************************************************}

\subsection{Shannon capacity of a graph}
\label{subsec-shannon-capacity-graph}

First we briefly review the
Shannon capacity of a graph which measures the zero-error capacity of a discrete memoryless point-to-point channel.
Throughout the paper all logarithms are taken to base 2.

Let $G = (V, E)$ be a graph with vertex set $V$ and edge set $E$. Two vertices $v_1, v_2$ are \textit{adjacent}, denoted as $v_1\sim v_2$, if there is an edge
between $v_1$ and $v_2$, i.e., $\{v_1,v_2\}\in E$.
An \textit{independent set} in $G$ is a subset of pairwise non-adjacent vertices.
A \textit{maximum independent set} is an independent set with the  largest possible number of vertices.
And this number is called the \textit{independence number} of $G$, denoted by $\alpha(G)$.
The \textit{complement} of graph $G$, denoted by $\overline{G}$, is a graph on the same vertices such that two distinct vertices of $\overline{G}$ are adjacent if and only if they are not adjacent in $G$.
We write $K_n$ and $\overline{K}_n$ as the complete and empty graph over $n$ vertices respectively.

Let $G = (V(G), E(G))$ and $H = (V(H), E(H))$ be two graphs. The \textit{strong product} (or normal product) $G\boxtimes H$ of the graphs $G$ and $H$ is a graph such that
\begin{itemize}
  \item[(1)] the vertex set of $G\boxtimes H$ is the Cartesian product $V(G)\times V(H)$;
  \item[(2)] two vertices $(u,u')$ and $(v,v')$ are adjacent if and only if one of the followings holds: (a) $u=v$ and $u'\sim v'$; (b) $u\sim v$ and $u'=v'$; (c) $u\sim v$ and $u'\sim v'$.
\end{itemize}
The $n$-fold strong product of graph $G$ with itself is denoted as $G^{n}$. The \textit{Shannon capacity of graph} $G$ was defined in \cite{Shannon1956} as
\begin{equation*}
\Theta(G)\triangleq \sup_{n} \frac{1}{n}\log \alpha(G^{n})=\lim_{n\to \infty} \frac{1}{n}\log \alpha(G^{n}),
\end{equation*}
where the limit exists by Fekete's lemma.

The \textit{disjoint union} $G\sqcup H$ of the graphs $G$ and $H$ is a graph such that $V(G\sqcup H)=V(G)\sqcup V(H)$ and $E(G\sqcup H)=E(G)\sqcup E(H)$.
A \textit{graph homomorphism} from $G$ to $H$, denoted as $G\to H$, is a mapping $\varphi:\, V(G)\to V(H)$ such that if $g_1\sim g_2$ in $G$, then $\varphi(g_1)\sim \varphi(g_2)$ in $H$.
Denote $G\preccurlyeq H$ if there is a graph homomorphism $\overline{G}\to \overline{H}$ from the complement of $G$ to the complement of $H$.

In \cite{Zuiddam}, Zuiddam introduced the asymptotic spectrum of graphs and provided a dual characterisation of the Shannon capacity of a graph
by applying Strassen's theory of asymptotic spectra, which includes the
Lov\'{a}sz theta number~\cite{Lovasz1979},
the fractional clique cover number,
the complement of the fractional orthogonal rank~\cite{CMR2014},
and the fractional Haemers' bound over any field
\cite{Bla2013, BC2018, Haemers1979} as specific elements of the asymptotic spectrum (also called spectral points).

\begin{theorem}[\cite{Zuiddam}]\label{thm-asym-spec}
Let $\mathcal{G}$ be a collection of graphs which is closed under the disjoint union $\sqcup$ and the strong product $\boxtimes$, and also contains the graph
with a single vertex $K_1$.
Define the \textit{asymptotic spectrum}  $\Delta(\mathcal{G})$ as the set of all mappings $\eta: \mathcal{G}\to \mathbb{R}_{\ge 0}$ such that for all
$G,H\in \mathcal{G}$
\begin{itemize}
  \item[(1)] if $G\preccurlyeq H$, then $\eta(G)\le \eta(H)$;
  \item[(2)] $\eta(G\sqcup H)=\eta(G)+\eta(H)$;
  \item[(3)] $\eta(G\boxtimes H)=\eta(G)\eta(H)$;
  \item[(4)] $\eta(K_1)=1$.
\end{itemize}
Then
$\Theta(G)=\min\limits_{\eta\in \Delta(\mathcal{G})} \log \eta(G)$.
In other words, $\min\limits_{\eta\in \Delta(\mathcal{G})} \eta(G)=2^{\Theta(G)}$ 
and $\alpha(G)\le \min\limits_{\eta\in \Delta(\mathcal{G})} \eta(G)$.
\end{theorem}

As remarked in \cite{Zuiddam}, $2^{\Theta(G)}$ is in general not an element of $\Delta(\mathcal{G})$.
In fact, $2^{\Theta(G)}$ is not additive under $\sqcup$ by a result of Alon~\cite{Alon98}, and also not multiplicative under $\boxtimes$ by a result of Haemers~\cite{Haemers1979}.
In Section \ref{subsec-outer-bd-spectrum-graphs}, to derive an outer bound for zero-error capacity of a DM-TWC, we will employ the multiplicativity of $\eta(G)$ for $\eta\in \Delta(\mathcal{G})$ under the $\boxtimes$ operation.

\subsection{Confusion graphs of channels}
In this subsection
we characterize the zero-error capacity region of a discrete memoryless point-to-point channel and a DM-TWC with respect to their corresponding confusion graphs.

A point-to-point channel consists of a finite input alphabet $\mathcal{X}$, a finite output alphabet $\mathcal{Y}$, and a conditional probability distribution $P_{Y|X}(y|x)$, where $x\in \mathcal{X}$, $y\in \mathcal{Y}$.
The channel is \textit{memoryless} in the sense that
$P_{Y_i|X^i,Y^{i-1}} (y_i|x^i,$ $y^{i-1}) = P_{Y|X} (y_i|x_i)$ for the $i$th channel use.
A transmitter would like to convey a message $m\in [2^{nR}]$ to a receiver over the channel.
To that end, the transmitter sends an input sequence $x^n\in \mathcal{X}^n$ using an encoding function $f: [2^{nR}]\to \mathcal{X}^n$ and
the receiver observes the channel outputs $y^n\in \mathcal{Y}^n$ and decodes using a decoding function $\phi: \mathcal{Y}^n\to [2^{nR}]$.
This pair $(f,\phi)$ is called an \textit{$(n,R)$ code} and such code is \textit{uniquely decodable} if $m=\phi(y^n)$ holds for any $m\in [2^{nR}]$ and correspondingly possible $y^n$.
A rate $R$ is \textit{achievable} if an $(n,R)$ uniquely decodable code exists for some $n$.
The \textit{zero-error capacity} of the channel is the supremum of all achievable rates.

In terms of the zero-error capacity, the channel $P_{Y|X}$ could be characterized by its \textit{confusion graph} $G$, whose vertex set is the input alphabet $\mathcal{X}$, and two vertices $x,x'\in \mathcal{X}$ are adjacent, denoted as $x\sim x'$, if and only if there exists $y\in \mathcal{Y}$ such that $P_{Y|X}(y|x)P_{Y|X}(y|x')>0$.
It is easy to verify that $\mathcal{C}$ is an $(n,R)$ uniquely decodable code if and only if $\mathcal{C}$ is an independent set of the graph $G^n$, which is the $n$-fold strong product of graph $G$. Consequently, the zero-error capacity of a point-to-point channel is equal to the Shannon capacity of its confusion graph $G$.

In the following we would like to characterize a DM-TWC using a collection of confusion graphs.
Indeed, we notice that in the two-way communication when Alice sends a letter $x_1\in \mathcal{X}_1$, the zero-error communication from Bob to Alice can be regarded as a point-to-point channel $P_{Y_1|X_1=x_1,X_2}$,
which corresponds to a confusion graph $G_{x_1}$ such that the vertex set is $\mathcal{X}_2$ and two vertices $x_2,x'_2\in \mathcal{X}_2$ are adjacent, also denoted as $x_2\overset{x_1}{\sim} x'_2$, if and only if there exists some $y_1\in \mathcal{Y}_1$ such that
\begin{equation*}
P_{Y_1|X_1,X_2}(y_1|x_1,x_2)P_{Y_1|X_1,X_2}(y_1|x_1,x'_2)>0,
\end{equation*}
where
\begin{equation*}
    P_{Y_1|X_1,X_2}(y_1|x_1,x_2)\triangleq\sum_{y_2\in \mathcal{Y}_2}P_{Y_1,Y_2|X_1,X_2}(y_1,y_2|x_1,x_2).
\end{equation*}
Symmetrically, when Bob sends a letter $x_2\in \mathcal{X}_2$, the zero-error communication from Alice to Bob could be characterized by a confusion graph $H_{x_2}$
such that the vertex set is $\mathcal{X}_1$ and two vertices $x_1,x'_1\in \mathcal{X}_1$ are adjacent, denoted as $x_1\overset{x_2}{\sim} x'_1$, if and only if there exists some $y_2\in \mathcal{Y}_2$ such that
\begin{equation*}
P_{Y_2|X_1,X_2}(y_2|x_1,x_2)P_{Y_2|X_1,X_2}(y_2|x'_1,x_2)>0,
\end{equation*}
where
\begin{equation*}
    P_{Y_2|X_1,X_2}(y_2|x_1,x_2)\triangleq\sum_{y_1\in \mathcal{Y}_1}P_{Y_1,Y_2|X_1,X_2}(y_1,y_2|x_1,x_2).
\end{equation*}
Based on the foregoing, a DM-TWC $P_{Y_1,Y_2|X_1,X_2}$ could be decomposed into a collection of discrete memoryless  point-to-point channels and thus be characterized by a sequence of confusion graphs, denoted by  $[G_1,\ldots,G_{|\mathcal{X}_1|};H_1,$ $\ldots,H_{|\mathcal{X}_2|}]$, where $V(G_1)=\cdots =V(G_{|\mathcal{X}_1|})=\mathcal{X}_2$ and $V(H_1)=\cdots =V(H_{|\mathcal{X}_2|})=\mathcal{X}_1$.
The following observation is immediate and useful.

\begin{prop}\label{prop-udc-is}
A codebook pair $(\mathcal{A},\mathcal{B})$ is uniquely decodable with respect to the channel $[G_1,\ldots,G_{|\mathcal{X}_1|};$ $ H_1,\ldots,$ $H_{|\mathcal{X}_2|}]$ if and only if
for any $a^n=(a_1,\ldots,a_n)\in \mathcal{A}$ and $b^n=(b_1,\ldots,b_n)\in \mathcal{B}$, we have that
$\mathcal{B}$ is an independent set of $G_{a_1}\boxtimes \cdots \boxtimes G_{a_n}$, and in the meanwhile $\mathcal{A}$ is an independent set of $H_{b_1}\boxtimes \cdots \boxtimes H_{b_n}$.
\end{prop}

Notice that in principle  $\mathscr{C}_{ze}(P_{Y_1,Y_2|X_1,X_2})$ only depends on its corresponding confusion graphs $[G_1,\ldots,G_{|\mathcal{X}_1|};H_1,$ $\ldots,H_{|\mathcal{X}_2|}]$.
Hence in the sequel we shall also use $\mathscr{C}_{ze}([G_1,\ldots,G_{|\mathcal{X}_1|};H_1,\ldots,H_{|\mathcal{X}_2|}])$ and $\mathscr{C}^{sum}_{ze}([G_1,\ldots,G_{|\mathcal{X}_1|};H_1,$ $\ldots,H_{|\mathcal{X}_2|}])$ to represent $\mathscr{C}_{ze}(P_{Y_1,Y_2|X_1,X_2})$ and $\mathscr{C}^{sum}_{ze}(P_{Y_1,Y_2|X_1,X_2})$ respectively.
For brevity, we also write the channel as $[\{G_i\};\{H_j\}]$ when it is clear from the context.
The following simple observation is analogues to the point-to-point case, and its proof is omitted.
\begin{prop}\label{prop-same-adjacency}
If $P_{Y_1,Y_2|X_1,X_2}$ and $Q_{Y_1,Y_2|X_1,X_2}$ have the same confusion graphs %adjacency
up to
some relabeling on input symbols,
%permutations of their input alphabets,
then $\mathscr{C}_{ze}(P_{Y_1,Y_2|X_1,X_2}) = \mathscr{C}_{ze}(Q_{Y_1,Y_2|X_1,X_2})$.
\end{prop}

This immediately implies
\begin{prop}
$\mathscr{C}_{ze}(P_{Y_1,Y_2|X_1,X_2})$ depends only on the conditional marginal distributions
$P_{Y_1|X_1,X_2}$ and $P_{Y_2|X_1,X_2}$.
\end{prop}

The strong product of two DM-TWCs $[G_1,\ldots,G_{|\mathcal{X}_1|};H_1,\ldots,H_{|\mathcal{X}_2|}]$ and $[G'_1,\ldots,G'_{|\mathcal{X}'_1|};H'_1,\ldots,H'_{|\mathcal{X}'_2|}]$, denoted by $[\{G_i\}; \{H_j\}]\boxtimes [\{G'_i\};\{H'_j\}]$, refers to a DM-TWC having input alphabets $\mathcal{X}_1\times \mathcal{X}'_1$ and $\mathcal{X}_2\times \mathcal{X}'_2$, as well as confusion graphs $$[\{G_i\boxtimes G'_{i'}:\, i\in \mathcal{X}_1, i'\in \mathcal{X}'_1 \}; \{H_j\boxtimes H'_{j'}:\, j\in \mathcal{X}_2, j'\in \mathcal{X}'_2 \}].$$    
Considering the zero-error sum-capacity with respect to the strong product, we have the lemma below.

\begin{lemma}\label{lemma-product-sumcapacity}
$\mathscr{C}^{sum}_{ze}\Big([\{G_i\}; \{H_j\}]\boxtimes [\{G'_i\};\{H'_j\}]\Big)\ge \mathscr{C}^{sum}_{ze}\Big([\{G_i\};\{H_j\}]\Big) + \mathscr{C}^{sum}_{ze}\Big([\{G'_i\};\{H'_j\}]\Big)$. 
\end{lemma}

\begin{proof}
To prove this lemma, it is sufficient to prove that for any $(n,R_1,R_2)$ (resp.\,$(n,R'_1,R'_2)$) uniquely decodable codebook pair $(\mathcal{A},\mathcal{B})$ (resp. $(\mathcal{A}',\mathcal{B}')$) with respect to channel $[\{G_i\}; \{H_j\}]$ (resp. $[\{G'_i\}; \{H'_j\}]$), we could have an $(n, R_1+R'_1,R_2+R'_2)$ uniquely decodable codebook pair for their product channel $[\{G_i\}; \{H_j\}]\boxtimes [\{G'_i\};\{H'_j\}]$. Indeed, let 
\begin{align*}
    \mathcal{A}^*&=\{((a_1,a'_1),\ldots,(a_n,a'_n)): a^n\in \mathcal{A}, a'^{n}\in \mathcal{A}'\},\\
    \mathcal{B}^*&=\{((b_1,b'_1),\ldots,(b_n,b'_n)): b^n\in \mathcal{B}, b'^{n}\in \mathcal{B}'\}.
\end{align*}
It is easy to verify that $(\mathcal{A}^*,\mathcal{B}^*)$ is uniquely decodable according to the product channel. Moreover, $|\mathcal{A}^*|=|\mathcal{A}||\mathcal{A}'|=2^{n(R_1+R'_1)}$ and $|\mathcal{B}^*|=|\mathcal{B}||\mathcal{B}'|=2^{n(R_2+R'_2)}$. The lemma follows. 
\end{proof}

\subsection{Dual graph homomorphisms}

In this subsection we study the properties of the zero-error capacity of a DM-TWC under graph homomorphisms, extending a similar analysis in the point-to-point channel~\cite{Shannon1956}.
Let $[\{G_i\}; \{H_j\}]$ and $[\{G'_i\};\{H'_j\}]$ be two sequences of confusion graphs corresponding to two DM-TWCs such that $V(G_i)=V(G)$,
$V(H_j)=V(H)$, $V(G'_i)=V(G')$ and $V(H'_j)=V(H')$.
A \textit{dual graph homomorphism} from
$[\{G_i\};\{H_j\}]$ to $[\{G'_i\};\{H'_j\}]$, denoted by $[\{G_i\};\{H_j\}]\to [\{G'_i\};\{H'_j\}]$,
is a pair of mappings $(\varphi,\psi)$, where $\varphi: V(H)\to V(H')$ and $\psi: V(G)\to V(G')$, such that
\begin{itemize}
  \item[(1)]
  if $v_1\sim v_2$ in $G_i$, then $\psi(v_1)\sim \psi(v_2)$ in $G'_{\varphi(i)}$; and
  \item[(2)]
  if $u_1\sim u_2$ in $H_j$, then $\varphi(u_1)\sim \varphi(u_2)$ in $H'_{\psi(j)}$.
\end{itemize}
It is easy to see that the dual graph homomorphism is a natural  generalization of the standard graph homomorphism of two graphs in the sense that they are both adjacency preserving.
Denote
$[\{G_i\};\{H_j\}]\preceq [\{G'_i\};\{H'_j\}]$ if there exists a dual graph homomorphism from
$[\{\overline{G}_i\};\{\overline{H}_j\}]$ to $[\{\overline{G}'_i\};\{\overline{H}'_j\}]$.
Now we have the following lemma.

\begin{lemma}
If $[\{G_i\};\{H_j\}]\preceq [\{G'_i\};\{H'_j\}]$, and $\overline{G}_i$ and $\overline{H}_j$ do not include self-loops, then
\begin{equation*}
\mathscr{C}_{ze}([\{G_i\};\{H_j\}])
\subseteq \mathscr{C}_{ze}([\{G'_i\};\{H'_j\}]).
\end{equation*}
\end{lemma}
\begin{proof}
Suppose $(\varphi,\psi): [\{\overline{G}_i\};\{\overline{H}_j\}] \to [\{\overline{G}'_i\};\{\overline{H}'_j\}]$
and
$(\mathcal{A},\mathcal{B})$ is a uniquely decodable codebook pair of length $n$ with respect to the DM-TWC
$[\{G_i\};\{H_j\}]$.
Let
\begin{align*}
\Phi(\mathcal{A})&=\{\varphi(a^n)=(\varphi(a_1),\ldots,\varphi(a_n)): a^n\in \mathcal{A}\},\\
\Psi(\mathcal{B})&=\{\psi(b^n)=(\psi(b_1),\ldots,\psi(b_n)): b^n\in \mathcal{B}\}.
\end{align*}
Now we would like to show that
$(\Phi(\mathcal{A}),\Psi(\mathcal{B}))$ is a uniquely decodable codebook pair for the DM-TWC
$[\{G'_i\};\{H'_j\}]$ such that $|\Phi(\mathcal{A})|=|\mathcal{A}|$ and  $|\Psi(\mathcal{B})|=|\mathcal{B}|$.

It suffices to show that for any distinct $a^n,\widetilde{a}^n\in \mathcal{A}$ and $b^n,\widetilde{b}^n\in \mathcal{B}$, we have
\begin{equation}\label{eq-nsim-pp}
\begin{split}
\varphi(a^n)&\nsim \varphi(\widetilde{a}^n)\ \text{in}\ H_{\psi(b_1)}\boxtimes \cdots \boxtimes H_{\psi(b_n)},  \\
 \psi(b^n)&\nsim \psi(\widetilde{b}^n)\ \text{in}\ G_{\varphi(a_1)}\boxtimes \cdots \boxtimes G_{\varphi(a_n)}.
\end{split}
\end{equation}
Indeed, since $(\mathcal{A},\mathcal{B})$ is a uniquely decodable codebook pair, there exist coordinates $i,j\in [n]$ such that
$a_i\nsim \widetilde{a}_i$ in $H_{b_i}$ and $b_j\nsim \widetilde{b}_j$ in $G_{a_j}$.
By the definition of $(\varphi,\psi)$, we have that
$\varphi(a_i)\nsim \varphi(\widetilde{a}_i)$ in $H_{\psi(b_i)}$ and $\psi(b_j)\nsim \psi(\widetilde{b}_j)$ in $G_{\varphi(a_j)}$,
implying (\ref{eq-nsim-pp}).
And then it is evident that  $|\Phi(\mathcal{A})|=|\mathcal{A}|$ and  $|\Psi(\mathcal{B})|=|\mathcal{B}|$.
The lemma follows by taking the union over all uniquely decodable codebook pairs $(\mathcal{A},\mathcal{B})$ with respect to
$[\{G_i\};\{H_j\}]$.
\end{proof}

\subsection{One-shot zero-error communication}
\label{subsec-oneshot}

%\ofer{verify this paragraph}
In this subsection we consider zero-error communication over a DM-TWC with only single channel use by the two parties (i.e., $n=1$). We refer to the associated set of achievable rate pairs as the \textit{one-shot zero-error capacity region}, and the associated sum-rate as the \textit{one-shot zero-error sum-capacity}. Recall that the one-shot zero-error capacity of a point-to-point channel is simply the logarithm of the independence number of its confusion graph; this quantity yields a lower bound on the (unrestricted) zero-error capacity of the channel, and also provides an infinite-letter expression for the capacity when evaluated over the product graph. It is therefore interesting to study the analogue of the independence number in the two-way case, which in particular would yield an inner bound on the zero-error capacity region of the DM-TWC. For simplicity of exposition, we will focus here on the one-shot zero-error sum-capacity only.

For convenience we define some notions first.
Let $[\{G_i\}; \{H_j\}]$ be a DM-TWC such that $V(G_i)=\mathcal{X}_2$ and $V(H_j)=\mathcal{X}_1$.
A pair $(S,T)$ of subsets $S\subseteq \mathcal{X}_1$ and $T\subseteq \mathcal{X}_2$ is called a \textit{dual clique pair} of the DM-TWC if $t\overset{s}{\sim}t'$ and $s\overset{t}{\sim}s'$ for any distinct $s,s'\in S$ and distinct $t,t'\in T$, that is,
$S$ is a clique in each
$H_t$ for $t\in T$, and $T$ is a clique in each $G_s$ for $s\in S$.
A pair $(S,T)$ of subsets $S\subseteq \mathcal{X}_1$ and $T\subseteq \mathcal{X}_2$ is called a \textit{dual independent pair} of the DM-TWC
if $T$ is an independent set of graph $G_s$ for each $s\in S$, and in the meanwhile, $S$ is an independent set of graph $H_t$ for each $t\in T$.
A \textit{maximum dual independent pair} is a dual independent pair $(S,T)$ with the largest possible size product $|S||T|$.
This product is called the \textit{independence product} of $[\{G_i\}; \{H_j\}]$,
denoted by $\pi(\{G_i\}; \{H_j\})$. According to the definition, 
the one-shot zero-error sum-capacity of the DM-TWC is $\log \pi(\{G_i\}; \{H_j\})$.
It is also readily seen that if two channels have the same confusion graphs up to some relabeling on input symbols, then they have the same collections of dual clique pairs and dual independent pairs, and hence the same one-shot zero-error sum-capacity.

For two graphs $G_1$ and $G_2$, let $G_1\cup G_2$ be the \textit{union} of $G_1$ and $G_2$ such that
$V(G_1\cup G_2)=V(G_1)\cup V(G_2)$ and $E(G_1\cup G_2)=E(G_1)\cup E(G_2)$.
Notice that the graph disjoint union $\sqcup$ in Section \ref{subsec-shannon-capacity-graph} is a special case of the union $\cup$, when the vertex sets of $G_1$ and $G_2$ are disjoint.
For notation convenience, in the rest of this subsection we let $|\mathcal{X}_1|=m_1$ and $|\mathcal{X}_2|=m_2$. The following simple observations are now in order.

\begin{prop}\label{prop-dual-indep}
Suppose $(S,T)$ is a dual independent pair of $[G_1,\ldots,G_{m_1}; H_1,\ldots,H_{m_2}]$.
\begin{itemize}
  \item[(1)]  If $|S|=1$, then $|T|\le \max\limits_{1\le i\le m_1} \alpha(G_i)$. The equality holds by taking $S=\{s\}$ and $T$ be a maximum independent set of $G_s$, where $s\in\arg\max_{1\le i\le m_1} \alpha(G_i)$.
  %\item[(2)]  If $|T|=1$, then $|S|\le \max\limits_{1\le j\le m_2} \alpha(H_j)$. The equality holds by taking $T=\{t\}$ and $S$ be a maximum independent set of $H_t$, where $t\in\arg\max_{1\le j\le m_2} \alpha(H_j)$. \ofer{isn't this redundant give the previous claim, by symmetry? If so it can be removed.}
  %\yujie{Yes, it is the same with (1) by symmetry.} \ofer{so let's remove this, unless oyu have  great reason not to.}
  \item[(2)]  $|S|\le \min\limits_{t\in T} \alpha(H_t)$.
  %, $|T|\le \min\limits_{s\in S} \alpha(G_s) $.
  \item[(3)]
 % $|S|\le \alpha(\bigcup_{t\in T}H_t)$,
 % $|T|\le \alpha(\bigcup_{s\in S} G_s)$.
  $S$ is an independent set of $\bigcup_{t\in T}H_t$.
  %, and $T$ is an independent set of $\bigcup_{s\in S} G_s$.
\end{itemize}
\end{prop}
\begin{proof}
The results follow directly from the definition of dual independent pairs.
\end{proof}

\begin{lemma} 
\label{lemma-properties}
Let $[G_1,\ldots,G_{m_1}; H_1,\ldots,H_{m_2}]$ be a DM-TWC and $G$, $H$ be graphs such that $V(G)=\mathcal{X}_2$, $V(H)=\mathcal{X}_1$. Then we have
\begin{itemize}
  \item[(1)]
  %\begin{equation*}\label{eq-pi-le-ge}
  $\max \Big\{\max\limits_{1\le i\le m_1} \alpha(G_i), \max\limits_{1\le j\le m_2} \alpha(H_j) \Big\}\le
  \pi(G_1,\ldots,G_{m_1}; H_1,\ldots,H_{m_2}) \le \max\limits_{1\le i\le m_1} \alpha(G_i)\cdot \max\limits_{1\le j\le m_2}  \alpha(H_j).$
  %\end{equation*}
  \item[(2)] $\pi(G,\ldots,G; H,\ldots,H)=\alpha(G) \alpha(H).$
  \item[(3)] $\pi(\overline{K}_{m_2},G,\ldots,G; \overline{K}_{m_1}, H,\ldots,H)=\max \{\alpha(G) \alpha(H), m_1, m_2\}.$
  \item[(4)] $\pi(G_1,\ldots,G_{m_1}; K_{m_1},\ldots,K_{m_1})=\max\limits_{1\le i\le m_1} \alpha(G_i)$.
\end{itemize}
\end{lemma}

\begin{proof}
(1)
The lower bound follows from Proposition \ref{prop-dual-indep} (1) and the symmetry of $S$ and $T$.
Also from Proposition \ref{prop-dual-indep} (2), we have
\begin{align*}
%\begin{split}
|S|&\le \min_{t\in T} \alpha(H_t)\le \max_{1\le j\le m_2} \alpha(H_j),\\
|T|&\le \min_{s\in S} \alpha(G_s) \le \max_{1\le i\le m_1} \alpha(G_i),
%\end{split}
\end{align*}
yielding the upper bound.

(2) From item (1), we have $\pi(G,\ldots,G; H,\ldots,H)\le \alpha(G) \alpha(H)$. The equality holds by taking $S$ and $T$ as the maximum independent sets of $H$ and $G$ respectively. 

(3) From items (1) and (2), we have $\pi(\overline{K}_{m_2},G,\ldots,G; \overline{K}_{m_1}, H,\ldots,H)\ge \max \{\alpha(G) \alpha(H), m_1, m_2\}$. On the other hand, suppose $(S,T)$ is a dual independent pair, then we have the following three cases.
(i) If $|S|=1$ then by Proposition \ref{prop-dual-indep} (1) we have $|S||T|\le m_2$.
(ii) If $|T|=1$, similar to case (i), we have $|S||T|\le m_1$.
(iii) If $|S|\ge 2$ and $|T|\ge 2$, then by Proposition \ref{prop-dual-indep} (2) we obtain $|S||T|\le \alpha(G) \alpha(H)$.
Thus $\pi(\overline{K}_{m_2},G,\ldots,G; \overline{K}_{m_1}, H,\ldots,H)\le \max \{\alpha(G) \alpha(H), m_1, m_2\}$.

(4) is a direct consequence of item (1).
The lemma follows.
\end{proof}

By graph homomorphisms we immediately have
\begin{prop}
If $[\{G_i\};\{H_j\}]\preceq [\{G'_i\};\{H'_j\}]$, then
\begin{equation*}
\pi(\{G_i\};\{H_j\})\leq \pi(\{G'_i\};\{H'_j\}).
\end{equation*}
\end{prop}

Next we shall provide an upper bound for $\pi(\{G_i\};\{H_j\})$ via a generalization of the Lov\'asz theta number~\cite{Lovasz1979}.
Let $\Gamma$ be an arbitrary $(m_1+m_2)\times (m_1+m_2)$ positive semi-definite matrix (i.e., $\Gamma \succeq 0$), and $\Gamma_{i,j}$ be its $(i,j)$th entry. Let 
$J$ be an $m_1\times m_2$ all-one matrix, and
$I_n$ be an $n\times n$ identity matrix.
For any matrices $A$ and $B$, denote $\langle A,B\rangle=\text{trace}(AB)$ and denote $A^{T}$ as the transpose of matrix $A$. 
Now define $\rho(\{G_i\},\{H_j\})$ as
\begin{align}
\text{maximize} \quad &\langle \left( \begin{array}{cc}
                           0 & J \\
                           J^T & 0
                         \end{array}
\right),\Gamma \rangle \nonumber \\
 \text{subject to}\quad &\langle \left( \begin{array}{cc}
                           I_{m_1} & 0 \\
                           0 & 0
                         \end{array}
\right),\Gamma \rangle =1, \nonumber \\
&\langle \left( \begin{array}{cc}
                           0 & 0 \\
                           0 & I_{m_2}
                         \end{array}
\right),\Gamma \rangle =1,\label{eq-sdp} \\
&\Gamma_{i,j+m_1}\cdot \Gamma_{i,k+m_1}=0,\ \ \ \forall\ i\in \mathcal{X}_1,\, j,k\in \mathcal{X}_2,\, j\ne k,\,
j\sim k \ \text{in}\ G_i \nonumber \\
&\Gamma_{i+m_1,j} \cdot \Gamma_{i+m_1,k}=0,\ \ \ \forall\ i\in \mathcal{X}_2,\, j,k\in \mathcal{X}_1,\, j\ne k,\,
j\sim k \ \text{in}\ H_i \nonumber \\
&\Gamma_{i,k+m_1} \cdot \Gamma_{j,k+m_1}=0,\ \ \ \forall\ i,j\in \mathcal{X}_1,\, k\in \mathcal{X}_2,\, i\ne j,\,
i\sim j \ \text{in}\ H_k \nonumber \\
&\Gamma_{i+m_1,k}\cdot \Gamma_{j+m_1,k}=0,\ \ \ \forall\ i,j\in \mathcal{X}_2,\, k\in \mathcal{X}_1,\, i\ne j,\,
i\sim j \ \text{in}\ G_k \nonumber \\
&\Gamma\succeq 0. \nonumber
\end{align}

\begin{lemma}\label{lemma-rho}
$\pi(\{G_i\},\{H_j\})\le \big(\frac{1}{2}\rho(\{G_i\},\{H_j\})\big)^2$.
\end{lemma}

\begin{proof}
Suppose $(S,T)$ with $S\subseteq \mathcal{X}_1$, $T\subseteq \mathcal{X}_2$ is a maximum dual independent pair such that $|S||T|=\pi(\{G_i\},\{H_j\})$.
For a number $m$ and a set $S$, denote
$m+S=\{m+s: s\in S\}$. Let $\Gamma$ be an $(m_1+m_2)\times (m_1+m_2)$ matrix such that
\begin{equation*}
    \Gamma_{i,j}=
    \begin{cases}
    \frac{1}{|S|}, &\text{if}\ i\in S,\, j\in S\\
    \frac{1}{\sqrt{|S||T|}}, &\text{if}\ i\in S,\, j\in m_1+T,\ \text{or}\ \, i\in m_1+T,\, j\in S\\
    \frac{1}{|T|}, &\text{if}\ i\in m_1+T,\, j\in m_1+T\\
    0, &\text{otherwise}.
    \end{cases}
\end{equation*}
Notice that for any vector $x^{m_1+m_2}=(x_1,\ldots,x_{m_1+m_2})$ we have 
$$x^{m_1+m_2}\cdot \Gamma\cdot (x^{m_1+m_2})^T
=\bigg(\frac{1}{\sqrt{|S|}}\sum_{i\in S} x_i + \frac{1}{\sqrt{|T|}}\sum_{j\in T} x_{m_1+j}\bigg)^2\ge 0.$$
This shows that $\Gamma$ is a positive semi-definite matrix satisfying the equality constraints in~\eqref{eq-sdp}.
Accordingly, $\Gamma$ is a feasible solution for program~\eqref{eq-sdp} and
\begin{align*}
%\begin{split}
    \rho(\{G_i\},\{H_j\})&\ge \langle \left( \begin{array}{cc}
                           0 & J \\
                           J^T & 0
                         \end{array}
\right),\Gamma \rangle \\
&=2\sqrt{|S||T|}\\
&=2\sqrt{\pi(\{G_i\},\{H_j\})},
\end{align*}
implying the result. This completes the proof.
\end{proof}

\section{Outer bounds}
\label{sec-outer-bd}

In this section we provide single-letter outer bounds for the non-adpative zero-error capacity region of the DM-TWC.
First in Section \ref{subsec-simple-outer-bd}, we present two simple outer bounds, one based on Shannon's vanishing-error non-adaptive capacity region and the other on a two-way analogue of the linear programming bound for point-to-point channels.
Next in Section \ref{subsec-LP-Shannon-outer-bd},
we combine the two bounds given in Section \ref{subsec-simple-outer-bd} and obtain an outer bound that is generally better than both.
Finally in Section \ref{subsec-outer-bd-spectrum-graphs},
we derive another single-letter outer bound via the asymptotic spectra of graphs.

\subsection{Simple bounds}
\label{subsec-simple-outer-bd}

It is trivial to see that Shannon's vanishing-error non-adaptive capacity region of the DM-TWC \cite[Theorem 3]{Shannon1961} contains its zero-error counterpart. First recall Shannon's bound in~\cite{Shannon1961}. 
\begin{lemma}[\cite{Shannon1961}]
The vanishing-error non-adaptive capacity region of a  DM-TWC $P_{Y_1,Y_2|X_1,X_2}$ is the convex hull of the set
\begin{equation*}
   \bigcup_{P_{X_1}, P_{X_2}} \{(R_1,R_2): R_1\ge 0, R_2\ge 0, R_1= I(X_1;Y_2|X_2), R_2= I(X_2;Y_1|X_1)\}
\end{equation*}
where the union is taken over all product input probability distributions $P_{X_1}\times P_{X_2}$. 
\end{lemma}

Together with Proposition \ref{prop-same-adjacency}, this immediately yields  the following outer bound.
\begin{lemma}\label{upperbd1}
$\mathscr{C}_{ze}(P_{Y_1,Y_2|X_1,X_2})$ is contained in
\begin{equation}\label{upbd-1-shannon}
\begin{split}
\bigcap_{Q_{Y_1,Y_2|X_1,X_2}}\bigcap_{0\le \lambda\le 1}
\{(R_1,&R_2): R_1\ge 0, R_2\ge 0, \lambda R_1 + (1-\lambda) R_2 \le \max_{P_{X_1},P_{X_2}}\epsilon (\lambda)\},
\end{split}
\end{equation}
where
\begin{equation}\label{def-eps-lambda}
\epsilon (\lambda)\triangleq
 \lambda I(X_1;Y_2|X_2)  + (1-\lambda) I(X_2;Y_1|X_1).
\end{equation}
The first intersection is taken over all DM-TWCs $Q_{Y_1,Y_2|X_1,X_2}$ with the same adjacency as $P_{Y_1,Y_2|X_1,X_2}$, and
the maximum is taken over all product input probability distributions $P_{X_1}\times P_{X_2}$.
\end{lemma}

\begin{remark}\rm
The bound (\ref{upbd-1-shannon}) can also be written in the standard form
\begin{equation*}
\begin{split}
\bigcap_{Q_{Y_1,Y_2|X_1,X_2}}\bigcup_{P_{X_1},P_{X_2}}
\big\{(R_1,R_2):\ & R_1\ge 0, R_2\ge 0, \\[-0.3cm]
& R_1\le I(X_1;Y_2|X_2),\\
&R_2\le I(X_2;Y_1|X_1)\big\}.
\end{split}
\end{equation*}
Here we prefer however to use the form (\ref{upbd-1-shannon}), for ease of comparison with forthcoming bounds.
\end{remark}

We now proceed to obtain a combinatorial outer bound. Recall that a dual clique pair of a DM-TWC is a pair $(S,T)$ of subsets $S\subseteq \mathcal{X}_1$ and $T\subseteq \mathcal{X}_2$ such that $t\overset{s}{\sim}t'$ and $s\overset{t}{\sim}s'$ for any distinct $s,s'\in S$ and distinct $t,t'\in T$. In the sequel, we adopt the convention that $0^0=1$.
\begin{lemma}\label{upperbd2}
$\mathscr{C}_{ze}(P_{Y_1,Y_2|X_1,X_2})$ is contained in
\begin{equation}\label{Eq-form-bd2}
%\begin{split}
\bigcap_{0\le \lambda\le 1} \{(R_1,R_2): R_1\ge 0, R_2\ge 0,  \lambda R_1+ (1-\lambda)R_2\le \max_{P_{X_1},P_{X_2}}-\log l(\lambda)\},
%\end{split}
\end{equation}
where %\\[-0.4cm]
\begin{equation}\label{def-l-lambda}
l(\lambda) \triangleq
\max_{S,T}
\left(\sum_{x_1\in S} P_{X_1}(x_1)\right)^{\lambda}   \left(\sum_{x_2\in T} P_{X_2}(x_2)\right)^{1-\lambda}
\end{equation}
and the maximum in \eqref{Eq-form-bd2} is taken over all the input probability distributions $P_{X_1}$ and $P_{X_2}$, and the maximum in  \eqref{def-l-lambda} is taken over all the dual clique pairs $(S,T)$ of $P_{Y_1,Y_2|X_1,X_2}$.
\end{lemma}
\begin{proof}
The argument is similar to the first part of the proof of Theorem \ref{upperbd-main}, and we omit it here.
\end{proof}

The following is a trivial corollary of Lemmas \ref{upperbd1} and \ref{upperbd2}.
\begin{coro}\label{coro-two-upbd}
$\mathscr{C}_{ze}(P_{Y_1,Y_2|X_1,X_2})$ is contained in
\begin{equation}
\begin{split}
\bigcap_{Q_{Y_1,Y_2|X_1,X_2}}\bigcap_{0\le \lambda\le 1}
\{(R_1,&R_2): R_1\ge 0, R_2\ge 0, \lambda R_1 + (1-\lambda) R_2 \le t(\lambda)\},
\end{split}
\end{equation}
where% \\[-0.35cm]
\begin{equation}\label{def-t-lambda}
t(\lambda)\triangleq\min\Big\{\max_{P_{X_1},P_{X_2}}\epsilon (\lambda), \max_{P_{X_1},P_{X_2}}-\log l(\lambda)\Big\}.
\end{equation}
\end{coro}

\subsection{An improved bound}
\label{subsec-LP-Shannon-outer-bd}

We now state a single-letter outer bound result, in which the order of the minimum and the maximum in~\eqref{def-t-lambda} is swapped. This generally yields a tighter outer bound due to the max-min inequality. In fact, our bound can be seen as a generalization of the one obtained by Holzman and K\"{o}rner for the binary multiplying channel~\cite{HK1995}, in which case the max-min is indeed strictly tighter than the min--max.

\begin{theorem}\label{upperbd-main}
$\mathscr{C}_{ze}(P_{Y_1,Y_2|X_1,X_2})$ is contained in
\begin{equation}\label{eq-upbd2}
\begin{split}
 \bigcap_{Q_{Y_1,Y_2|X_1,X_2}} \bigcap_{0\le \lambda \le 1}\{(R_1,&R_2): R_1\ge 0, R_2\ge 0, \lambda R_1 + (1-\lambda) R_2 \le \theta(\lambda)\},
\end{split}
\end{equation}
where
\begin{equation}\label{def-theta-lambda}
%\begin{split}
\theta(\lambda)\triangleq
\max_{P_{X_1},P_{X_2}}
\min \{
\epsilon (\lambda), -\log l(\lambda)\}.
%\end{split}
\end{equation}
The first intersection is taken over all DM-TWCs $Q_{Y_1,Y_2|X_1,X_2}$ with the same adjacency as $P_{Y_1,Y_2|X_1,X_2}$, and
the maximum is taken over all product input probability distributions $P_{X_1}\times P_{X_2}$.
\end{theorem}

\begin{proof}
The intersection over all $Q_{Y_1,Y_2|X_1,X_2}$ follows from Proposition \ref{prop-same-adjacency}. Hence without loss of generality, we prove that for $P_{Y_1,Y_2|X_1,X_2}$, each achievable rate pair $(R_1,R_2)$ satisfies $\lambda R_1 + (1-\lambda) R_2 \le \theta(\lambda)$, where $0\le \lambda \le 1$.

To that end, for each uniquely decodable codebook pair $(\mathcal{A},\mathcal{B})$ of length $n$, we will show that
\begin{equation}\label{to-proof}
|\mathcal{A}|^{\lambda}|\mathcal{B}|^{1-\lambda}\le \kappa 2^{n\theta(\lambda)}
\end{equation}
by induction on $n$, where $\kappa$ is a constant independent of $n$.
For the base case $n=1$, one could take subsets $\mathcal{A}\subseteq \mathcal{X}_1$, $\mathcal{B}\subseteq \mathcal{X}_2$ such that for any distinct $a,a'\in \mathcal{A}$ and distinct $b,b'\in \mathcal{B}$, we have $a\overset{b}{\nsim}a'$ and $b\overset{a}{\nsim}b'$. Clearly, $|\mathcal{A}||\mathcal{B}|\le |\mathcal{X}_1||\mathcal{X}_2|$ and (\ref{to-proof}) follows by taking $\kappa$ sufficiently large. Assume that~\eqref{to-proof} holds for every length $n'\le n-1$, and let us proceed to prove for length $n$. Suppose $(\mathcal{A},\mathcal{B})\subseteq \mathcal{X}^n_1\times \mathcal{X}^n_2$ is a uniquely decodable codebook pair of length $n$.  For a vector $x^n$, let $x^{n\setminus i}\triangleq (x_1,\ldots,x_{i-1},x_{i+1},\ldots,x_n)$
be its projection over all coordinates not equal to $i$. For each coordinate $1\le i\le n$ and each $x_1\in \mathcal{X}_1$, $x_2\in  \mathcal{X}_2$, let
\begin{equation}\label{A_iB_i}
\begin{split}
\mathcal{A}_i(x_1)&\triangleq\{a^{n\setminus i}:\ a^n\in \mathcal{A},\ a_i=x_1\}, \\
\mathcal{B}_i(x_2)&\triangleq\{b^{n\setminus i}:\ b^n\in \mathcal{B},\ b_i=x_2\}
\end{split}
\end{equation}
be the projections of each codebook in the $i$th coordinate. Define the distributions induced by these projections over $\mathcal{X}_1$ and $\mathcal{X}_2$ respectively to be
\begin{equation}\label{p_iq_i}
P^i_{X_1}(x_1)\triangleq\frac{|\mathcal{A}_i(x_1)|}{|\cal{A}|}, \qquad
P^i_{X_2}(x_2)\triangleq\frac{|\mathcal{B}_i(x_2)|}{|\cal{B}|}.
\end{equation}
Furthermore, for any two subsets $S\subseteq \mathcal{X}_1$ and $T\subseteq \mathcal{X}_2$, define the codebooks induced by the unions over $S$ and $T$ of the respective projected codebooks over the $i$th coordinate to be
\begin{equation}\label{eq:union_of_proj_code}
    \mathcal{A}_i(S)\triangleq\bigcup_{x_1\in S}\mathcal{A}_i(x_1),\quad \mathcal{B}_i(T)\triangleq\bigcup_{x_2\in T}\mathcal{B}_i(x_2).
\end{equation}

Note that if $(S,T)$ is a dual clique pair such that $\mathcal{A}_i(S)\ne \emptyset$ and $ \mathcal{B}_i(T)\ne \emptyset$, then
the unions in~\eqref{eq:union_of_proj_code} are disjoint, otherwise it contradicts the assumption that $(\mathcal{A},\mathcal{B})$ is uniquely decodable. Hence
\begin{align}\label{eq:dual_disjoint_union}
    |\mathcal{A}_i(S)|=\sum_{x_1\in S}|\mathcal{A}_i(x_1)|, \quad
               |\mathcal{B}_i(T)|=\sum_{x_2\in T}|\mathcal{B}_i(x_2)|
\end{align}
and also, for any $i\in[n]$ it must hold that $(\mathcal{A}_i(S),\mathcal{B}_i(T))$ is a uniquely decodable codebook pair of length $n-1$.

Now, if there exist a dual clique pair $(S,T)$ and a coordinate $1\le i\le n$ such that
\begin{equation}\label{assump-cond}
    \left(\sum_{s\in S}P^i_{X_1}(s)\right)^{\lambda}\left(\sum_{t\in T}P^i_{X_2}(t)\right)^{1-\lambda} \ge 2^{-\theta(\lambda)},
\end{equation}
then
\begin{align}
\label{eq:induction_first}|\mathcal{A}|^{\lambda}  |\mathcal{B}|^{1-\lambda}    &=\left(\frac{|\mathcal{A}_i(S)|}{\sum_{s\in S}P^i_{X_1}(s)}\right)^{\lambda}  \cdot \left(\frac{|\mathcal{B}_i(T)|}{\sum_{t\in T}P^i_{X_2}(t)})\right)^{1-\lambda}\\
\label{eq:induction_second} &
\le \frac{\kappa 2^{(n-1)\theta(\lambda)}}{2^{-\theta(\lambda)} } \\
&=\kappa 2^{n\theta(\lambda)},
\end{align}
where~\eqref{eq:induction_first} follows from~\eqref{eq:dual_disjoint_union}; and~\eqref{eq:induction_second} follows from the inductive hypothesis, assumption~\eqref{assump-cond}, and the fact that $(\mathcal{A}_i(S),\mathcal{B}_i(T))$ is a uniquely decodable codebook pair of length $n-1$. We conclude that~\eqref{to-proof} holds under condition~\eqref{assump-cond}.

Assume now that condition~\eqref{assump-cond} is not satisfied, that is, %that
\begin{equation}\label{assump-cond-2}
\max_{i\in[n]}\max_{S,T}\left(\sum_{s\in S}P^i_{X_1}(s)\right)^{\lambda}
\left(\sum_{t\in T}P^i_{X_2}(t)\right)^{1-\lambda}
 < 2^{-\theta(\lambda)}.
\end{equation}
Let $A^n$ and $B^n$ be codewords chosen from $\mathcal{A}$ and $\mathcal{B}$ respectively, uniformly at random, and let $Y^n_1$, $Y^n_2$ be the corresponding channel outputs. Since $(\mathcal{A},\mathcal{B})$ is a uniquely decodable codebook pair of length $n$, it must be that
\begin{equation}\label{=logAB}
\begin{split}
    \log |\mathcal{A}|&=I(Y^n_2;A^n|B^n),\\
    \log |\mathcal{B}|&=I(Y^n_1;B^n|A^n).
\end{split}
\end{equation}

On the other hand, we have
\begin{align}\label{le-nI1}
 I(Y_1^n;B^n|A^n)&=H(Y_1^n|A^n)-H(Y_1^n|A^n,B^n)\\
%&= \sum_{i=1}^n H(Y_{1,i}|Y_{1,1},\ldots,Y_{1,i-1},A^n)   - \sum_{i=1}^n   H(Y_{1,i}|Y_{1,1},\ldots,Y_{1,i-1},A^n,B^n)\\
\label{eq:chain_rule}&= \sum_{i=1}^n H(Y_{1,i}|Y_{1,1},\ldots,Y_{1,i-1},A^n) - \sum_{i=1}^n   H(Y_{1,i}|A_i,B_i)\\
\label{eq:cond_red_ent}&\le  \sum_{i=1}^n   H(Y_{1,i}|A_i)  - \sum_{i=1}^n   H(Y_{1,i}|A_i,B_i)\\
&= \sum_{i=1}^n I(Y_{1,i};B_i|A_i),
\end{align}
where~\eqref{eq:chain_rule} follows from the entropy chain rule and the memorylessness of the channel,
and~\eqref{eq:cond_red_ent} follows from the fact that conditioning reduces entropy. Similarly,
\begin{equation}\label{le-nI2}
I(Y_2^n;A^n|B^n) \le \sum_{i=1}^n I(Y_{2,i};A_i|B_i).
\end{equation}
Combining~\eqref{assump-cond-2}-\eqref{le-nI2}, we obtain
\begin{align}\label{Eq8}
&\log |\mathcal{A}|^{\lambda}|\mathcal{B}|^{1-\lambda} \nonumber \\
&= \lambda \log |\mathcal{A}| + (1-\lambda) \log |\mathcal{B}| \nonumber  \\
&\le \sum_{i=1}^n   \lambda I(Y_{2,i};A_i|B_i) + (1-\lambda) I(Y_{1,i};B_i|A_i) \nonumber  \\
&\le \max_{P_{X_1},P_{X_2},\atop l(\lambda)< 2^{-\theta(\lambda)}}
%n\cdot \epsilon (\lambda),
n[\lambda I(Y_2;X_1|X_2) +(1-\lambda) I(Y_1;X_2|X_1)]\nonumber  \\
&= \max_{P_{X_1},P_{X_2},\atop l(\lambda)< 2^{-\theta(\lambda)}}
n\cdot \epsilon (\lambda),
\end{align}
where $\epsilon(\lambda)$ and $l(\lambda)$ are defined in~\eqref{def-eps-lambda} and~\eqref{def-l-lambda} respectively,
and the maximum is taken over all product input probability distributions $P_{X_1}\times P_{X_2}$ such that
$l(\lambda)< 2^{-\theta(\lambda)}$, following condition~\eqref{assump-cond-2}.

By the definition of $\theta(\lambda)$, we have
\begin{align}
\theta(\lambda)&=\max_{P_{X_1},P_{X_2}}\min \{\epsilon (\lambda), -\log l(\lambda)\}\\
 \label{theta=}&\ge
\max_{P_{X_1},P_{X_2},\atop l(\lambda) < 2^{-\theta(\lambda)} }
\min \{\epsilon (\lambda), -\log l(\lambda)\}.
\end{align}
Note that for any input distributions $P_{X_1},P_{X_2}$ such that $l(\lambda) < 2^{-\theta(\lambda)}$, we have
\begin{equation}\label{>logtheta}
-\log l(\lambda)>\theta(\lambda).
\end{equation}
Combining (\ref{theta=}) and (\ref{>logtheta}), we obtain
\begin{equation}\label{eq11}
\max_{P_{X_1},P_{X_2},\atop l(\lambda)< 2^{-\theta(\lambda)}}
\epsilon (\lambda) \le \theta(\lambda).
\end{equation}
Substituting~\eqref{eq11} into~\eqref{Eq8}, we have $\log|\mathcal{A}|^{\lambda}|\mathcal{B}|^{1-\lambda}\le n \theta(\lambda)$,
completing the proof.
\end{proof}

We remark that Theorem \ref{upperbd-main} immediately implies, in particular, the following upper bound on the zero-error capacity of the point-to-point discrete memoryless channel.
\begin{coro}\label{bd-one-way}
The zero-error capacity of the discrete memoryless channel $P_{Y|X}$ is upper bounded by
\begin{equation*}
    \min_{Q_{Y|X}}\, \max_{P_X}\, \min \Big\{I(X;Y), -\log \max_{C} \sum_{x\in C} P_X(x)\Big\}.
\end{equation*}
The first minimum is taken over all the $Q_{Y|X}$ having the same confusion graph as $P_{Y|X}$, the first maximum is taken over all the input distributions $P_X$, and the second maximum is taken over all the cliques $C$ of the confusion graph of the channel.
\end{coro}

It is less obvious that the upper bound in Corollary \ref{bd-one-way} in fact coincides with the linear programming bound on the zero-error capacity of a point-to-point discrete memoryless channel in \cite{Shannon1956}. This follows from a conjecture proposed by Shannon \cite{Shannon1956} that has later been proved by Ahlswede \cite{Ahlswede1973}. That is, $$\min_{Q_{Y|X}}\, \max_{P_X} I(X;Y)=\max_{P_X} \Big\{-\log \max_{C} \sum_{x\in C} P_X(x)\Big\}$$
for any point-to-point discrete memoryless channel $P_{Y|X}$. In other words, 
this means that in the point-to-point case, Corollary~\ref{coro-two-upbd} yields exactly the same bound as Theorem~\ref{upperbd-main}. However, this is not the case in general for the DM-TWC.
For example, recall that Holzman and K\"{o}rner~\cite{HK1995} derived the bound in Theorem~\ref{upperbd-main} in the special case of the (deterministic) binary multiplying channel (using  $\lambda=0.5$) and numerically showed that it is strictly better than what can be obtained from Corollary~\ref{coro-two-upbd}. Next we give another example showing that Theorem~\ref{upperbd-main} outperforms Corollary~\ref{coro-two-upbd} for a noisy (i.e., non-deterministic) DM-TWC as well.

\begin{example}\label{example-channel}\rm
Let $\mathcal{X}_1 = \{ 0, 1, 2 \},    \mathcal{X}_2=  \mathcal{Y}_1 =  \mathcal{Y}_2 =  \{ 0, 1 \} $, and
the conditional probability distribution $P_{Y_1,Y_2|X_1,X_2}$ be
\begin{center}
   \begin{tabular}{c|*{6}c}%{c|c|c|c|c|c|c}
  \hline
  % after \\: \hline or \cline{col1-col2} \cline{col3-col4} ...
  \backslashbox{ $y_1y_2$} {$x_1x_2$}& $00$ & $01$ & $10$ & $11$ & $20$ & $21$ \\ \hline
  $00$ & $1$ & $1$ & $0$ & $0$ & $0$ & $0$\\
  $01$ & $0$ & $0$ & $0$ & $0$ & $1$ & $0$\\
  $10$ & $0$ & $0$ & $\delta$ & $0$ & $0$ & $1$\\
  $11$ & $0$ & $0$ & $1-\delta$ & $1$ & $0$ & $0$\\
  \hline
\end{tabular}
\end{center}
where $\delta\in(0,1)$. Corollary~\ref{coro-two-upbd} gives the upper bound 
\begin{equation*}
    R_1+R_2\le \min \Big\{ \max_{P_{X_1},P_{X_2}} \epsilon^*,  \max_{P_{X_1},P_{X_2}} -\log l^* \Big\} \approx 1.2933,
\end{equation*}
where 
\begin{align*}
    \epsilon^* &= I(X_1;Y_2|X_2) + I(X_2;Y_1|X_1) \\
    &= P_{X_1}(2)\cdot h(P_{X_2}(0)) + P_{X_2}(0) \cdot h(P_{X_1}(0) + \delta \cdot P_{X_1}(1)) - P_{X_1}(1)\cdot P_{X_2}(0) \cdot h(\delta) + P_{X_2}(1)\cdot h(P_{X_1}(1)),\\
    l^*&= \max_{S,T}
\left(\sum_{x_1\in S} P_{X_1}(x_1)\right)  \left(\sum_{x_2\in T} P_{X_2}(x_2)\right)\\
&=\max \big\{P_{X_1}(0), P_{X_1}(1), P_{X_2}(0)\cdot (P_{X_1}(0)+ P_{X_1}(1)), P_{X_2}(0)\cdot (P_{X_1}(1)+P_{X_1}(2)), \\
&\qquad \qquad P_{X_2}(1)\cdot (P_{X_1}(0)+ P_{X_1}(2))\big\},
\end{align*}
%$R_1+R_2\le 1.2933$, 
and $h(x)=-x\log x -(1-x)\log(1-x)$. Whereas Theorem~\ref{upperbd-main} yields $$R_1+R_2\le \max\limits_{P_{X_1},P_{X_2}}\min \{ \epsilon^*, -\log l^* \}\approx 1.2910.$$
\end{example}

\subsection{An outer bound via Shannon capacity of a graph}
%the asymptotic spectrum of graphs}
%{***************************************************************}
\label{subsec-outer-bd-spectrum-graphs}

Based on Lemma~\ref{lemma-properties} and the Shannon capacity of a graph, we directly have the following bound. 
\begin{lemma}\label{lemma-upper-bound-via-Shannon-graph}
\begin{align*}
    \mathscr{C}^{sum}_{ze}([G_1,\ldots,G_{|\mathcal{X}_1|};H_1,\ldots,H_{|\mathcal{X}_2|}])
     \le \max_{x_1\in \mathcal{X}_1, x_2\in \mathcal{X}_2} \Theta(G_{x_1}) + \Theta(H_{x_2}).
\end{align*}    
\end{lemma}

It is worth noting that the above bound could be optimal in the sense that when all $G_i=G$ and $H_j=H$, it is easily verified that  $\mathscr{C}^{sum}_{ze}([G,\ldots,G;H,\ldots,H])=\Theta(G)+\Theta(H)$. 
However the bound in Lemma~\ref{lemma-upper-bound-via-Shannon-graph} is not tight in general. 
Later in Section~\ref{sec-certain-TWC}, we will improve the bound in Lemma~\ref{lemma-upper-bound-via-Shannon-graph} for certain scenarios and show that the improved bound (in Theorem~\ref{thm-G_K}) could outperform Theorem~\ref{upperbd-main} (see Example~\ref{example-thm3-outperform-thm2}) and be achieved in special cases (see Theorem~\ref{thm-construction-KG-KK}).

%************************************************************************************
\section{Inner bounds}\label{sec-inner-bd}

In this section we present two inner bounds for the non-adaptive zero-error capacity region of the DM-TWC, one based on random coding and the other on linear codes.

\subsection{Random coding}
The random coding for DM-TWC is standard and generalizes a known bound by Shannon for the one-way case~\cite{Shannon1956}. To obtain the random coding inner bound,
%The proof of the random coding lower bound is similar to the vanishing-error case in \cite{Shannon1961}, where
the following lemma in \cite{Shannon1961} is required.

\begin{lemma}[\cite{Shannon1961}]\label{shancomblemma}
Let $X$ be a random variable taking values in $[N]$, and $\{f_i:[N]\to \mathbb{R}_+\}_{i\in [d]}$ be a collection of nonnegative functions. Then there exists $x\in[N]$ such that $f_i(x) \leq d \cdot \mathbb{E}[f_i(X)] $ for all $i\in[d]$.
%Suppose we have a set of objects $B_1,B_2,\ldots,B_n$ with associated probabilities $P_1,P_2,\ldots,P_n$, and a number of numerically valued properties (functions)
%of the objects $f_1,f_2,\ldots,f_d$. These are all nonnegative, $f_i(B_j)\ge 0$, and $\sum_j[P_jf_i(B_j)]=A_i$, $i=1,2,\ldots,d$.
%Then there exists an object $B_p\in \{B_1,\ldots,B_n\}$ such that
%\begin{equation}
%    f_i(B_p) \le dA_i,\ \ \ \forall\, 1\le i\le d.
%\end{equation}
\end{lemma}
%\ofer{I think this lemma is trying to say something very simple/standard, but it is very confusing. Maybe we can even remove it and just explain inside the proof of the theorem where it is used?}
%\yujie{This lemma is trying to say that there exists a random object whose several properties $f_i$ could be upper bounded by the corresponding averaging value. When $d=1$ (only single property), it is just the simple pigeonhole principle.}

%\ofer{Why pigeonhole? More like method of expectation. Anyway, here is a simpler way to write this lemma, please verify and use if correct. Let $X$ be a r.v. taking values in $[n]$, and $\{f_i:[n]\to \mathbb{R}_+\}_{i\in [d]}$ a collection of nonnegative functions. Then there exists $x\in[n]$ such that $f_i(x) \leq d \cdot \mathbb{E} f_i(X) $ for all $i\in[d]$.}
%\yujie{Nice. Let's use your simpler way. Thanks.}

%{\color{blue}START-3. Please skim this proof, thanks. ***************************************************************}
\begin{theorem}\label{lowerbd1}
$\mathscr{C}_{ze}(P_{Y_1,Y_2|X_1,X_2})$ contains the region
\begin{equation}\label{lowerR}
\begin{split}
 \bigcup_{P_{X_1},P_{X_2}} %\bigcap_{0\le \lambda \le 1}
\Big\{(R_1,R_2):\ &R_1\ge 0, R_2\ge 0, \\
& R_1\le -\frac{1}{2} \log
\sum_{x_1\overset{x_2}{\sim}x_1' \;\vee
\; x_1=x_1' , \atop x_1,x'_1\in \mathcal{X}_1, x_2\in \mathcal{X}_2} P_{X_1}(x_1)P_{X_1}(x'_1)P_{X_2}(x_2),\\
&R_2\le  -\frac{1}{2} \log
\sum_{x_2\overset{x_1}{\sim}x_2' \;\vee
\;x_2=x_2', \atop x_1\in \mathcal{X}_1, x_2,x'_2\in \mathcal{X}_2}  P_{X_1}(x_1)P_{X_2}(x_2)P_{X_2}(x'_2)
\Big\}
\end{split}
\end{equation}
where the union is taken over all input distributions $P_{X_1}$, $P_{X_2}$.
\end{theorem}

\begin{proof}
First we randomly construct
%an ensemble of 
a codebook pair $(\mathcal{A},\mathcal{B})$ such that $\mathcal{A}$ (resp.\,$\mathcal{B}$) consists of $M_1$ (resp.$\, M_2$) words and each word is generated i.i.d. according to a probability distribution $P_{X_1}$ (resp.$\, P_{X_2}$).
We are going to show that there exists a pair $(\mathcal{A},\mathcal{B})$ which is uniquely decodable after some modifications.

For codebook pair $(\mathcal{A},\mathcal{B})$, %in the ensemble,
a word $a^n\in \mathcal{A}$ is called \textit{bad}, if there exist two words
$b^n,\widetilde{b}^n\in \mathcal{B}$
that are either equal or adjacent in $G_{a_1}\boxtimes \cdots\boxtimes G_{a_n}$.
For any particular words $a^n\in \mathcal{A}$, $b^n,\widetilde{b}^n\in \mathcal{B}$ and coordinate $i\in [n]$, the probability that
$b_i\sim \widetilde{b}_i$ in $G_{a_i}$ is upper bounded by
\begin{equation*}
\sum_{x_2\overset{x_1}{\sim}x_2' \;\vee
\;x_2=x_2', \atop x_1\in \mathcal{X}_1, x_2,x'_2\in \mathcal{X}_2}  P_{X_1}(x_1)P_{X_2}(x_2)P_{X_2}(x'_2).
\end{equation*}
Since all the coordinates are independent,
the probability that $b^n\sim \widetilde{b}^n$ in $G_{a_1}\boxtimes \cdots\boxtimes G_{a_n}$ is at most
\begin{equation}\label{badpr}
\bigg(\sum_{x_2\overset{x_1}{\sim}x_2' \;\vee
\;x_2=x_2', \atop x_1\in \mathcal{X}_1, x_2,x'_2\in \mathcal{X}_2}  P_{X_1}(x_1)P_{X_2}(x_2)P_{X_2}(x'_2)\bigg)^n.
\end{equation}
Denote by $\text{Bad}(a^n)$ the number of $2$-subsets $\{b^n,\widetilde{b}^n\}\subseteq \mathcal{B}$ such that
$b^n\sim \widetilde{b}^n$ in $G_{a_1}\boxtimes \cdots\boxtimes G_{a_n}$. Then
%$p(y^n_1|a^n,b^n)p(y^n_1|a^n,\widetilde{b}^n)>0$ for some $y^n_1$. Then
%
\begin{align*}
%\begin{split}
\text{Pr}\{a^n \ \text{is bad}\}&= \text{Pr}\{\text{Bad}(a^n)\ge 1\}\\
&\le \mathbb{E}[\text{Bad}(a^n)]\\
&\le \binom{M_2}{2}\bigg(\sum_{x_2\overset{x_1}{\sim}x_2' \;\vee
\;x_2=x_2', \atop x_1\in \mathcal{X}_1, x_2,x'_2\in \mathcal{X}_2}  P_{X_1}(x_1)P_{X_2}(x_2)P_{X_2}(x'_2)\bigg)^n,
%\end{split}
\end{align*}
where the first inequality is by Markov's inequality,
and the second inequality follows from (\ref{badpr}) and the linearity of expectation.
Similarly,
a word $b^n\in \mathcal{B}$ is called \textit{bad}, if there exist two words
$a^n,\widetilde{a}^n\in \mathcal{A}$
that are equal or adjacent in $H_{b_1}\boxtimes \cdots\boxtimes H_{b_n}$, and we have
\begin{equation*}
\text{Pr}\{b^n \ \text{is bad}\} \le \binom{M_1}{2}\bigg( \sum_{x_1\overset{x_2}{\sim}x_1' \;\vee
\; x_1=x_1' , \atop x_1,x'_1\in \mathcal{X}_1, x_2\in \mathcal{X}_2} P_{X_1}(x_1)P_{X_1}(x'_1)P_{X_2}(x_2)   \bigg)^n.
\end{equation*}

Let $f_1(\mathcal{A},\mathcal{B})$, $f_2(\mathcal{A},\mathcal{B})$ be the number of bad words in $\mathcal{A}$ and $\mathcal{B}$ respectively. Then we have
\begin{align}%\label{expM1}
%\begin{split}
\mathbb{E}[f_1(\mathcal{A},\mathcal{B})]&\le M_1\binom{M_2}{2}\bigg(\sum_{x_2\overset{x_1}{\sim}x_2' \;\vee\;x_2=x_2', \atop x_1\in \mathcal{X}_1, x_2,x'_2\in \mathcal{X}_2}  P_{X_1}(x_1)P_{X_2}(x_2)P_{X_2}(x'_2)\bigg)^n,\label{expM1}\\
\mathbb{E}[f_2(\mathcal{A},\mathcal{B})]&\le M_2 \binom{M_1}{2}\bigg( \sum_{x_1\overset{x_2}{\sim}x_1' \;\vee\; x_1=x_1' , \atop x_1,x'_1\in \mathcal{X}_1, x_2\in \mathcal{X}_2} P_{X_1}(x_1)P_{X_1}(x'_1)P_{X_2}(x_2) \bigg)^n.
%\end{split}
\end{align}
By Lemma \ref{shancomblemma}, %\ofer{if you use my shortened version, might need to modify notation a bit in order for it to fit} \yujie{modified. What do you think about the current notation?},
there exists a pair $(\mathcal{A}^*,\mathcal{B}^*)$ %in the ensemble 
such that
%the number of bad words in $\mathcal{A}^*$ and $\mathcal{B}^*$, denoted by $U_1(\mathcal{A}^*)$ and $U_2(\mathcal{B}^*)$, such that
\begin{equation}\label{badM1}
f_1(\mathcal{A}^*,\mathcal{B}^*)\le 2\mathbb{E}[f_1(\mathcal{A},\mathcal{B})],\quad f_2(\mathcal{A}^*,\mathcal{B}^*)\le 2\mathbb{E}[f_2(\mathcal{A},\mathcal{B})].
\end{equation}
Remove all the bad words in $\mathcal{A}^*$ and $\mathcal{B}^*$ respectively,
yielding a codebook pair $(\mathcal{A}',\mathcal{B}')$ such that
\begin{equation}\label{A'B'}
|\mathcal{A}'|= M_1-f_1(\mathcal{A}^*,\mathcal{B}^*)\ \ \ \text{and}\ \ \ |\mathcal{B}'|= M_2-f_2(\mathcal{A}^*,\mathcal{B}^*).
\end{equation}
It is readily seen that $(\mathcal{A}',\mathcal{B}')$ is a uniquely decodable codebook pair. % Next we explore the rate of the code.

Now let
\begin{align}%\label{letM}
M_1&=(1-\xi_1)^{\frac{n}{2}} \bigg(\sum_{x_1\overset{x_2}{\sim}x_1' \;\vee\; x_1=x_1' , \atop x_1,x'_1\in \mathcal{X}_1, x_2\in \mathcal{X}_2} P_{X_1}(x_1)P_{X_1}(x'_1)P_{X_2}(x_2)\bigg)^{-\frac{n}{2}},\\
M_2&=(1-\xi_2)^{\frac{n}{2}} \bigg(\sum_{x_2\overset{x_1}{\sim}x_2' \;\vee\;x_2=x_2', \atop x_1\in \mathcal{X}_1, x_2,x'_2\in \mathcal{X}_2}  P_{X_1}(x_1)P_{X_2}(x_2)P_{X_2}(x'_2)\bigg)^{-\frac{n}{2}},\label{letM}
\end{align}
where $\xi_1,\xi_2$ are arbitrarily small positive numbers.
Combining (\ref{expM1})-(\ref{letM}),
we obtain
\begin{equation*}
\begin{split}
|\mathcal{A}'|&\ge (1-(1-\xi_2)^n) (1-\xi_1)^{\frac{n}{2}}  \bigg(\sum_{x_1\overset{x_2}{\sim}x_1' \;\vee\; x_1=x_1' , \atop x_1,x'_1\in \mathcal{X}_1, x_2\in \mathcal{X}_2} P_{X_1}(x_1)P_{X_1}(x'_1)P_{X_2}(x_2)\bigg)^{-\frac{n}{2}},\\
|\mathcal{B}'|&\ge (1-(1-\xi_1)^n) (1-\xi_2)^{\frac{n}{2}}  \bigg(\sum_{x_2\overset{x_1}{\sim}x_2' \;\vee\;x_2=x_2', \atop x_1\in \mathcal{X}_1, x_2,x'_2\in \mathcal{X}_2}  P_{X_1}(x_1)P_{X_2}(x_2)P_{X_2}(x'_2)\bigg)^{-\frac{n}{2}}.
\end{split}
\end{equation*}
Since $\xi_1,\xi_2$ are arbitrarily small, by taking $n$ sufficiently large,
we could have an $(n,R_1,R_2)$ uniquely decodable codebook pair arbitrarily close to (\ref{lowerR}),
as desired.
\end{proof}

\subsection{Linear codes}
In this subsection we present a construction of uniquely decodable codes via linear codes,   which generalizes the result for binary multiplying channel~\cite{Tolhuizen}. Let us introduce some notation first. Suppose $D$ is a set of letters, $x^n$ and $y^n$ are vectors of length $n$, and $\mathcal{C}$ is a collection of vectors of length $n$.  Let
\begin{equation}\label{eq-def-ind}
    \text{ind}_{D}(x^n)\triangleq \{1\le i\le n: x_i\in D\}
\end{equation}
denote the collection of indices where $x_i\in D$. For $I\subseteq [n]$ let $y^n|_{I}$ denote the vector obtained from $y^n$ by projecting onto the coordinates in $I$, and denote
\begin{equation*}
    \mathcal{C}_{|I}\triangleq\{c^n|_{I}: c^n\in \mathcal{C}\}.
\end{equation*}

Let  $P_{Y_1,Y_2|X_1,X_2}$ be a DM-TWC. We say that $x_1\in \mathcal{X}_1$ is a \textit{detecting symbol}, if $x_2\overset{x_1}{\not\sim}x'_2$ for any distinct $x_2,x'_2\in \mathcal{X}_2$. A detecting symbol $x_2\in \mathcal{X}_2$ is defined analogously. Let $D_1\subseteq \mathcal{X}_1$ and $D_2\subseteq \mathcal{X}_2$ denote the sets of all detecting symbols in $\mathcal{X}_1$ and $\mathcal{X}_2$ respectively.
A vector $a^n\in \mathcal{X}^n_1$
is called a \textit{detecting vector} for $\mathcal{B}\subseteq \mathcal{X}^n_2$ if
\begin{equation}\label{def-detect-vector}
    \left|\mathcal{B}_{|\text{ind}_{D_1}(a^n)}\right|=|\mathcal{B}|.
\end{equation}
Similarly,
a vector $b^n\in \mathcal{X}^n_2$ is a \textit{detecting vector} for  $\mathcal{A}\subseteq\mathcal{X}^n_1$ if
\begin{equation}
    |\mathcal{A}_{|\text{ind}_{D_2}(b^n)}|=|\mathcal{A}|.
\end{equation}
The following claim is immediate.

\begin{prop}\label{prop-sufficient}
Let $\mathcal{A}\subseteq \mathcal{X}^n_1$,  $\mathcal{B}\subseteq \mathcal{X}^n_2$.
If each $a^n\in \mathcal{A}$ is a detecting vector for $\mathcal{B}$ and
each $b^n\in \mathcal{B}$ is a detecting vector for $\mathcal{A}$,
then $(\mathcal{A},\mathcal{B})$ is a uniquely decodable codebook pair.
\end{prop}

Proposition \ref{prop-sufficient} provides a sufficient condition for a uniquely decodable code, which is not always necessary (see Example~\ref{example}). Nevertheless, this sufficient condition furnishes us with a way of constructing uniquely decodable codes by employing linear codes.
\begin{example}\label{example}\rm
Suppose $\mathcal{X}_1=\{a_0,a_1,a_2\}$, $\mathcal{X}_2=\{b_0,b_1\}$ such that $D_1=\{a_0,a_1,a_2\}$, $D_2=\{b_1\}$, and  $a_0\overset{b_0}{\sim}a_1$, $a_0\overset{b_0}{\sim}a_2$, $a_1\overset{b_0}{\nsim}a_2$. Let $\mathcal{A}=\{a_0a_0a_0, a_1a_1a_1, a_0a_1a_2\}$ and $\mathcal{B}=\{b_0b_1b_0\}$.
It is easy to verify that $(\mathcal{A},\mathcal{B})$ is a uniquely decodable codebook pair. However, $\text{ind}_{D_2}(b_0b_1b_0)=\{2\}$ and $|\mathcal{A}_{|\{2\}}|=|\{a_0,a_1\}|=2<|\mathcal{A}|=3$, implying that $b_0b_1b_0$ is not a detecting vector for $\mathcal{A}$.
\end{example}

If we assume that $|\mathcal{X}_1|=q_1$ and $|\mathcal{X}_2|=q_2$, where $q_1,q_2$ are prime powers, %hence
then we think of the alphabets as $\mathbb{F}_{q_1}$ and $\mathbb{F}_{q_2}$ respectively.
The following theorem gives an inner bound on the capacity region, which is a generalization of the Tolhuizen's construction for the Blackwell's multiplying channel~\cite{Tolhuizen}.

\begin{theorem}\label{lowerlinear}
Let $P_{Y_1,Y_2|X_1,X_2}$ be a DM-TWC with input alphabet sizes $|\mathcal{X}_1|=q_1$, $|\mathcal{X}_2|=q_2$,
where $q_1$, $q_2$ are prime powers. If $\mathcal{X}_1$ and $\mathcal{X}_2$ contain $\tau_1$ and $\tau_2$ detecting symbols respectively,
then $\mathscr{C}_{ze}(P_{Y_1,Y_2|X_1,X_2})$ contains the region
\begin{equation}\label{eq-inner-1}
\begin{split}
 \bigcup_{0\le \alpha,\beta\le 1} %\bigcap_{0\le \lambda \le 1}
\big\{(R_1,R_2):\ &R_1\ge 0, R_2\ge 0, \\
& R_1\le 
h(\alpha)+ \alpha \log\tau_2 +  (1-\alpha) \log(q_2-\tau_2) -(1-\beta)\log q_2,\\
&R_2\le  
h(\beta)+  \beta \log\tau_1 +   (1-\beta) \log(q_1-\tau_1)- (1-\alpha)\log q_1
\big\}
\end{split}
\end{equation}
where $h(x)\triangleq-x\log x-(1-x) \log (1-x)$ is the binary entropy function.
\end{theorem}

To prove this theorem, we need the following lemma. The case that $q_1=q_2=2$ and $\tau=1$ was proved in \cite[Theorem 3]{Tolhuizen}.
Here Lemma \ref{randomlinear} follows a similar argument.

\begin{lemma}\label{randomlinear}
Let $q$, $q'$ be prime powers, $n,k$ be positive integers such that $1\le k\le n$, and $D\subseteq \mathbb{F}_{q'}$ with cardinality  $|D|=\tau$.
Then there exists a pair $(\mathcal{C}, \Upsilon(\mathcal{C}))$ satisfying that
\begin{itemize}
    \item[(1)] $\mathcal{C}$ is a $q$-ary $[n,k]$ linear code;
    \item[(2)] $\Upsilon(\mathcal{C})\subseteq \mathbb{F}^n_{q'}$ such that %\\%[-0.3cm]
\begin{equation*}
    |\Upsilon(\mathcal{C})|\ge \binom{n}{k}\tau^k(q'-\tau)^{n-k}\prod_{i=1}^{\infty} (1-q^{-i});%\\%[-0.1cm]
\end{equation*}
    \item[(3)]  for each $x^n\in \Upsilon(\mathcal{C})$, we have $|\text{ind}_D(x^n)|=k$ and  $\left|\mathcal{C}_{|\text{ind}_D(x^n)}\right|=|\mathcal{C}|$.
\end{itemize}
\end{lemma}

%{\color{blue}START-4. Please skim this proof, thanks. ***************************************************************}
\begin{proof}
Let $A$ be a $k\times n$ matrix of full rank over $\mathbb{F}_q$, then
$\mathcal{C}(A)\triangleq\{y^kA: y^k\in \mathbb{F}^k_q\}$ is a $q$-ary $[n,k]$ linear code generated by $A$.
Recall that for every $x^n\in \mathbb{F}^n_{q'}$, $\text{ind}_D(x^n)=\{i\in [n]:\, x_i\in D\}$ as in (\ref{eq-def-ind}).
Denote
\begin{equation*}
    \Upsilon(\mathcal{C}(A))\triangleq \Big\{x^n\in \mathbb{F}^n_{q'}:\,
    |\text{ind}_{D}(x^n)|=k,\, |\mathcal{C}_{|\text{ind}_{D}(x^n)}|=|\mathcal{C}|\Big\}.
\end{equation*}
Let $A|_{\text{ind}_{D}(x^n)}$ denote the $k\times |\text{ind}_{D}(x^n)|$ submatrix of $A$ with columns indexed by $\text{ind}_{D}(x^n)$.
It is easy to see that $|\mathcal{C}_{|\text{ind}_{D}(x^n)}|=|\mathcal{C}|$ is equivalent to
$\text{rank}(A|_{\text{ind}_{D}(x^n)})=k$.
Denote
\begin{equation}\label{eq-def-P}
    \mathcal{P}\triangleq\Big\{(A,x^n):\, A\in \mathbb{F}^{k\times n}_q,\,x^n\in \mathbb{F}^n_{q'},\,
    |\text{ind}_{D}(x^n)|=k,\, \text{rank}(A|_{\text{ind}_{D}(x^n)})=k\Big\},
\end{equation}
and let us proceed by double counting the cardinality of $\mathcal{P}$.

On the one hand,
the number of vectors $x^n\in \mathbb{F}^n_{q'}$ such that $|\text{ind}_{D}(x^n)|=k$ is $\binom{n}{k}\tau^k(q'-\tau)^{n-k}$.
For each such $x^n$, there are $q^{k(n-k)}I_q(k)$ corresponding $k\times n$ matrices $A\in \mathbb{F}^{k\times n}_q$
such that $\text{rank}(A|_{\text{ind}_{D}(x^n)})=k$, where
$I_q(k)=\prod^{k-1}_{i=0} (q^k-q^i)$ is the number of $k\times k$ invertible matrices over $\mathbb{F}_q$, see \cite[Lemma 3]{Tolhuizen}.
Hence we have
\begin{equation*}
|\mathcal{P}|=\binom{n}{k}\tau^k(q'-\tau)^{n-k}q^{k(n-k)}I_q(k).
\end{equation*}

On the other hand,
the number of matrices $A\in \mathbb{F}^{k\times n}_q$ is $q^{nk}$.
By (\ref{eq-def-P}) and the pigeonhole principle, there exist a matrix $A^*\in \mathbb{F}^{k\times n}_q$ and a corresponding code $\mathcal{C}(A^*)$ such that
$|\Upsilon(\mathcal{C}(A))|\ge |\mathcal{P}|/{q^{nk}}$. Letting $\mathcal{C}=\mathcal{C}(A^*)$, the lemma follows.
\end{proof}
%{\color{blue}END-2. Please skim this proof, thanks. ***************************************************************}

\begin{proof}[Proof of Theorem \ref{lowerlinear}]
For $i=1,2$,
let us identify $\mathcal{X}_i$ with $\mathbb{F}_{q_i}$, and let the respective sets of all detecting symbols be $D_i\subseteq \mathbb{F}_{q_i}$ with $|D_i|=\tau_i$.

To prove the existence of a uniquely decodable codebook pair based on Proposition \ref{prop-sufficient}, we first use Lemma \ref{randomlinear} to find two ``one-sided'' uniquely decodable linear codebook pairs, and then combine them to the desired codebook pair by employing their cosets in $\mathbb{F}^n_{q_1}$ and $\mathbb{F}^n_{q_2}$.

First, %for each $i\in \{1,2\}$,
letting $q=q_1$, $q'=q_2$, $G=D_2$ and $\tau=\tau_2$ in Lemma \ref{randomlinear}, we have a pair $(\mathcal{C}_1, \Upsilon(\mathcal{C}_1))$ satisfying that
$\mathcal{C}_1$ is a $q_1$-ary $[n,k_1]$ linear code and  $\Upsilon(\mathcal{C}_1)\subseteq \mathbb{F}^n_{q_2}$ such that
\begin{equation}\label{eq-size-one}
    |\Upsilon(\mathcal{C}_1)|\ge \binom{n}{k_1}{\tau_2}^{k_1}(q_2-\tau_2)^{n-k_1}\prod_{i=1}^{\infty} (1-q_{1}^{-i}).
\end{equation}
Similarly,
letting $q=q_2$, $q'=q_1$, $G=D_1$ and $\tau=\tau_1$ in Lemma \ref{randomlinear}, we have a pair $(\mathcal{C}_2, \Upsilon(\mathcal{C}_2))$ satisfying that
$\mathcal{C}_2$ is a $q_2$-ary $[n,k_2]$ linear code and  $\Upsilon(\mathcal{C}_2)\subseteq \mathbb{F}^n_{q_1}$ such that
\begin{equation}\label{eq-size-one-2}
    |\Upsilon(\mathcal{C}_2)|\ge \binom{n}{k_2}{\tau_1}^{k_2}(q_1-\tau_1)^{n-k_2}\prod_{i=1}^{\infty} (1-q_{2}^{-i}).
\end{equation}
The property $(3)$ in Lemma \ref{randomlinear} implies that each $x^n\in \Upsilon(\mathcal{C}_i)$ is a detecting vector for $\mathcal{C}_i$ for $i=1,2$.
Note that if $\Xi(\mathcal{C}_i)\subseteq \mathbb{F}^n_{q_i}$ is a coset of $\mathcal{C}_i$, then each
$x^n\in \Upsilon(\mathcal{C}_i)$ is also a detecting vector for $\Xi(\mathcal{C}_i)$.

Now we are going to combine the two pairs $(\mathcal{C}_1, \Upsilon(\mathcal{C}_1))$ and $(\mathcal{C}_2, \Upsilon(\mathcal{C}_2))$.
Since
%$\mathcal{C}_1$, $\mathcal{C}_2$ are considered in the same space $\mathbb{F}^n_q$
$\mathcal{C}_i$ has $q_i^{n-k_i}$ cosets, then by the pigeonhole principle there exists coset $\Xi(\mathcal{C}_i)$  of $\mathcal{C}_i$ such that
%containing a subset $\mathcal{A}\subseteq D(\mathcal{C}_1)$ of detecting vectors for $\mathcal{C}_1$, and
\begin{equation}\label{pigeon}
\begin{split}
    \mathcal{A} &\triangleq \Upsilon(\mathcal{C}_1)\cap \Xi(\mathcal{C}_2), \quad |\mathcal{A}| \ge \frac{| \Upsilon(\mathcal{C}_1)|}{q_2^{n-k_2}}, \\
    \mathcal{B} &\triangleq \Upsilon(\mathcal{C}_2)\cap \Xi(\mathcal{C}_1), \quad |\mathcal{B}| \ge \frac{| \Upsilon(\mathcal{C}_2)|}{q_1^{n-k_1}} .
\end{split}
\end{equation}

We now notice that each vector in $\mathcal{A}$ (resp.$\, \mathcal{B}$) is a detecting vector for $\mathcal{B}$ (resp.$\, \mathcal{A}$), hence by Proposition \ref{prop-sufficient}
$(\mathcal{A},\mathcal{B})$ is a uniquely decodable codebook pair. Moreover, for fixed $q_1$, $q_2$, we have
\begin{equation*}
\begin{split}
\frac {\log |\mathcal{A}|} {n}
\ge &\ h(k_1/n) + (k_1/n)\log\tau_2 + (1- k_1/n) \log(q_2-\tau_2) - (1-k_2/n)\log q_2\ -\ O(1/n),
\end{split}
\end{equation*}
\begin{equation*}
\begin{split}
\frac {\log |\mathcal{B}|} {n}
\ge &\ h(k_2/n) + (k_2/n)\log\tau_1 + (1- k_2/n) \log(q_1-\tau_1) - (1-k_1/n)\log q_1\ -\ O(1/n),
\end{split}
\end{equation*}
which follows from (\ref{eq-size-one}), (\ref{eq-size-one-2}) and (\ref{pigeon}).
Letting $\alpha=k_1/n\in [0,1]$, $\beta=k_2/n\in [0,1]$, we obtain 
\begin{align*}
    R_1 &= \lim\limits_{n\to \infty} \frac {\log |\mathcal{A}|} {n} \ge 
    h(\alpha)+ \alpha \log\tau_2 +  (1-\alpha) \log(q_2-\tau_2) -(1-\beta)\log q_2,\\
    R_2 &= \lim\limits_{n\to \infty} \frac {\log |\mathcal{B}|} {n} \ge 
    h(\beta)+  \beta \log\tau_1 +   (1-\beta) \log(q_1-\tau_1)- (1-\alpha)\log q_1.
\end{align*}
Therefore~\eqref{eq-inner-1} follows, as desired.
\end{proof}

We note that for any DM-TWC, one could only exploit part of input symbols  $\mathcal{X}'_1\subseteq \mathcal{X}_1$, $\mathcal{X}'_2\subseteq \mathcal{X}_2$ to meet the requirements in Theorem \ref{lowerlinear}. Hence we have the following general bound.
\begin{coro}\label{coro-lowerbd}
Let $P_{Y_1,Y_2|X_1,X_2}$ be a DM-TWC with input alphabets $\mathcal{X}_1$, $\mathcal{X}_2$.
Then $\mathscr{C}_{ze}(P_{Y_1,Y_2|X_1,X_2})$ contains the region
\begin{equation*}
\begin{split}
\bigcup_{\mathcal{X}'_1\subseteq \mathcal{X}_1,\atop \mathcal{X}'_2\subseteq \mathcal{X}_2}
 \bigcup_{0\le \alpha,\beta\le 1} %\bigcap_{0\le \lambda \le 1}
\big\{(R_1,R_2):\ &R_1\ge 0, R_2\ge 0, \\
& R_1\le 
h(\alpha)+ \alpha \log\tau'_2 +  (1-\alpha) \log(q'_2-\tau'_2) -(1-\beta)\log q'_2,\\
&R_2\le  
h(\beta)+  \beta \log\tau'_1 +   (1-\beta) \log(q'_1-\tau'_1)- (1-\alpha)\log q'_1
\big\}
% \bigcap_{0\le \lambda \le 1} \{(R_1,&R_2): R_1\ge 0, R_2\ge 0, \lambda R_1+(1-\lambda) R_2 \le L^{(\mathcal{X}'_1,\mathcal{X}'_2)}(\lambda)\},
\end{split}
\end{equation*}
where the first union is taken over all $\mathcal{X}'_1\subseteq \mathcal{X}_1$, $\mathcal{X}'_2\subseteq \mathcal{X}_2$ such that $|\mathcal{X}'_1|$ and $|\mathcal{X}'_2|$ are prime powers, and
contain $\tau'_1$ and $\tau'_2$ detecting symbols respectively.
\end{coro}

Notice that the region~\eqref{eq-inner-1}
relies on the number $q_1$, $q_2$ of symbols being used and the corresponding numbers $\tau_1,\tau_2$ of detecting symbols. It is thus possible that using only a smaller subset of channel inputs would yield higher achievable rates (when using our linear coding strategy) than those obtained by using larger subsets. For Example~\ref{example-channel}, Corollary~\ref{coro-lowerbd} shows that a lower bound on the maximum sum-rate $R_1+R_2$ is $1.17$, which is better than the random coding lower bound $1.0907$.

\section{Certain types of DM-TWC}
\label{sec-certain-TWC}
%{ ***************************************************************}

In this section we consider the DM-TWC in the scenario that the communication in one direction is stable (in particular, noiseless).
First we briefly review the probabilistic refinement of the Shannon capacity of a graph in Section~\ref{sec:5.1}.  
Then in Section \ref{subsec-certain-type-outer-bd} we provide an outer bound on the zero-error capacity region via the asymptotic spectrum of graphs. In Section \ref{subsec-construction-certain-TWC} we present explicit constructions that attain our outer bound in certain special cases.

\subsection{Probabilistic refinement of the Shannon capacity of a graph}
\label{sec:5.1}

We first recall some basic notions and results from the method-of-types.
Let $x^n\in \mathcal{X}^n$ be a sequence and $N(a|x^n)$ be the number of times that $a\in \mathcal{X}$ appears in sequence $x^n$.
The \textit{type} $P_{x^n}$ of $x^n$ is the relative proportion of each symbol in $\mathcal{X}$, that is, $P_{x^n}(a)\triangleq \frac{N(a|x^n)}{n}$ for all $a\in \mathcal{X}$. Let $\mathcal{P}_n$ denote the collection of all possible types of sequences of length $n$.
For every $P\in \mathcal{P}_n$, the \textit{type class} $T^n(P)$ of $P$ is the set of sequences of type $P$, that is, $T^n(P)\triangleq \{x^n: P_{x^n}=P\}$.
The \textit{$\epsilon$-typical set} of $P$ is $$T^n_{\epsilon}(P)\triangleq \{x^n\in \mathcal{X}^n: |P_{x^n}(a)-P(a)|<\epsilon,\, \forall\, a\in \mathcal{X}\}.$$
The \textit{joint type} $P_{x^n,y^n}$ of a pair of sequences $(x^n, y^n)$ is the relative proportion of occurrences of each pair of symbols of $\mathcal{X}\times \mathcal{Y}$, i.e.,
$P_{x^n,y^n}\triangleq \frac{N(a,b|x^n,y^n)}{n}$ for all $a\in \mathcal{X}$ and $b\in \mathcal{Y}$.
By the Bayes' rule, the \textit{conditional type} $P_{x^n|y^n}$ is defined as
\begin{equation*}
P_{x^n|y^n}(a,b)\triangleq \frac{N(a,b|x^n,y^n)}{N(b|y^n)}=\frac{P_{x^n,y^n}(a,b)}{P_{y^n}(b)}.
\end{equation*}
\begin{lemma}[\cite{Csiszar-Korner}]\label{lemma-types}
$|\mathcal{P}_n|\le (n+1)^{|\mathcal{X}|}$.
\end{lemma}
\begin{lemma}[\cite{Csiszar-Korner}]\label{lemma-types-2}
$\forall\, P\in \mathcal{P}_n$, we have
$\frac{2^{nH(X)}}{(n+1)^{|\mathcal{X}|}}\leq |T^n(P)|\le 2^{nH(X)}$, where $H(X)=-\sum_{x\in \mathcal{X}} P(x)\log P(x)$.
\end{lemma}

In \cite{CK1981}, Csisz\'{a}r and K\"{o}rner introduced the probabilistic refinement of the Shannon capacity of a graph, imposing
that the independent set consists of sequences of the (asymptotically) same type.
Let
$G^n_{\epsilon}[P]$ denote the subgraph of $G^n$ induced by $T^n_{\epsilon}(P)$.
The Shannon capacity of graph $G$ relative to $P$ is defined as
\begin{equation*}
\Theta(G,P)\triangleq \lim_{\epsilon\to 0}\limsup_{n\to \infty}\frac{1}{n}\log \alpha(G^n_{\epsilon}[P]).
\end{equation*}
Let
$G^n [P]$ denote the subgraph of $G^n$ induced by $T^n (P)$.
Then it is readily seen that 
$$
\limsup_{n\to \infty}\frac{1}{n}\log \alpha(G^n [P])\le \Theta(G,P).
$$
For each $\eta\in \Delta(\mathcal{G})$, define
\begin{equation}\label{eq-lim-eta-P}
\widehat{\eta}(G,P)\triangleq \lim_{\epsilon\to 0}\limsup_{n\to \infty} \frac{1}{n}\log \eta(G^n_{\epsilon}[P]).
\end{equation}
If $G=\overline{K}_n$, then according to Lemma~\ref{lemma-types-2}, we have 
\begin{equation}\label{eq-observation}
    \widehat{\eta}(\overline{K}_n,P)=H(X)
\end{equation}
for any $\eta\in \Delta(\mathcal{G})$. Very recently, Vrana~\cite{Vrana2019} proved the following results on $\widehat{\eta}(G,P)$.
%
%(Remove: that the limit in $\epsilon$ exists \ofer{right?}, and moreover:)

\begin{lemma}[\cite{Vrana2019}]\label{lemma-Vrana}
The limit in (\ref{eq-lim-eta-P}) exists and
\begin{itemize}
 \item[(1)]  $\Theta(G,P)=\min\limits_{\eta\in \Delta(\mathcal{G})}\widehat{\eta}(G,P)$;

    \item[(2)] $\log \eta(G)=\max\limits_{P}\, \widehat{\eta}(G,P)$ for $\eta\in \Delta(\mathcal{G})$.
\end{itemize}
\end{lemma}

According to Lemma \ref{lemma-types}, it is easily seen that
\begin{equation*}
    \Theta(G)=\max_{P}\, \Theta(G,P).
\end{equation*}
Here we would like to mention that the probabilistic refinement of the Lov\'asz theta number was introduced and investigated by Marton in \cite{Mar93} via a non-asymptotic formula, and the probabilistic refinement of the
fractional clique cover number was studied relating to the graph entropy in \cite{Korner73}.

\subsection{An outer bound via the asymptotic spectrum of graphs}
%Shannon capacity of a graph}
\label{subsec-certain-type-outer-bd}

In this subsection, we first show an outer bound for the case when all $\{H_j\}$ are the same, namely, $H_j=H$ for all $j\in \mathcal{X}_2$.  

\begin{theorem}\label{thm-G_K}
$\mathscr{C}_{ze}([G_1,\ldots,G_{|\mathcal{X}_1|};
H,\ldots,H])$ is contained in the region
\begin{align*}
\Big\{(R_1,R_2): R_1\ge 0, R_2\ge 0,
R_1+R_2\le
\max_{P_{X_1}}
 \sum_{x_1\in \mathcal{X}_1} P_{X_1}(x_1)\Theta(G_{x_1}) 
+ 
\Theta(H,P_{X_1})
\Big\}.
\end{align*}
\end{theorem}

\begin{proof}
Suppose $(\mathcal{A}, \mathcal{B})\subseteq \mathcal{X}^n_1\times \mathcal{X}^n_2$ is a uniquely decodable codebook pair of length $n$.
For any $a^n\in \mathcal{A}$ and $b^n\in \mathcal{B}$, let $P_{a^n,b^n}$ denote the joint type of the pair $(a^n,b^n)$ and
\begin{equation*}
J^n(P_{X_1,X_2})\triangleq \{(a^n,b^n): a^n\in \mathcal{A},b^n\in \mathcal{B}, P_{a^n,b^n}=P_{X_1,X_2}\}.
\end{equation*}
By Lemma \ref{lemma-types}, there are at most $(n+1)^{|\mathcal{X}_1||\mathcal{X}_2|}$ different joint types over $(\mathcal{A},\mathcal{B})$.
Thus by the pigeonhole principle,
there exists one joint type $P^*_{X_1,X_2}$ such that
\begin{equation}\label{eq-exist-J}
|J^n(P^*_{X_1,X_2})|\ge \frac{|\mathcal{A}||\mathcal{B}|}{(n+1)^{|\mathcal{X}_1||\mathcal{X}_2|}}.
\end{equation}
Notice that for each $(a^n,b^n)\in J^n(P^*_{X_1,X_2})$, we have
\begin{equation*}
\begin{split}
P_{a^n}&=P^*_{X_1}=\sum_{x_2\in \mathcal{X}_2} P^*_{X_1,X_2=x_2},\\
P_{b^n}&=P^*_{X_2}=\sum_{x_1\in \mathcal{X}_1} P^*_{X_1=x_1,X_2}.
\end{split}
\end{equation*}

Now we are going to upper bound the cardinality of $J^n(P^*_{X_1,X_2})$.
Let $\mathcal{A}^*$ (resp. $\mathcal{B}^*$) denote the collection of $a^n\in \mathcal{A}$ (resp. $b^n\in \mathcal{B}$) that appears in $J^n(P^*_{X_1,X_2})$, that is, there exists $b^n\in \mathcal{B}$ (resp. $a^n\in \mathcal{A}$) such that $P_{a^n,b^n}=P^*_{X_1,X_2}$.
Then we immediately have
\begin{equation}\label{eq-J-leq}
|J^n(P^*_{X_1,X_2})|\le |\mathcal{A}^*||\mathcal{B}^*|.
\end{equation}
Thus we turn to upper bound the cardinalities of $\mathcal{A}^*$ and $\mathcal{B}^*$.
Since $(\mathcal{A},\mathcal{B})$ is uniquely decodable,
by Proposition \ref{prop-udc-is},
for any $a^n\in \mathcal{A}^*$
we must have that $\mathcal{B}^*$ is an independent set of $G_{a_1}\boxtimes G_{a_2}\boxtimes \cdots \boxtimes G_{a_n}$. Accordingly, 
\begin{align}\label{eq-A_star}
|\mathcal{B}^*|
& \leq 
\alpha\Big(G^{nP^*_{X_1}(1)}_{1}\boxtimes G^{nP^*_{X_1}(2)}_{2}\boxtimes \cdots \boxtimes G^{nP^*_{X_1}(|\mathcal{X}_1|)}_{|\mathcal{X}_1|} \Big).
\end{align}
Also, 
for $b^n\in \mathcal{B}^*$, we notice that
$\mathcal{A}^*$ is an independent set of $H^n$ with type $P^*_{X_1}$. To be precise, we have
\begin{align}\label{eq-B_star}
|\mathcal{A}^*|
& \leq 
\alpha\Big(H^n[P^*_{X_1}]\Big)
\end{align}

Therefore we have
\begin{align}
&\limsup_{n\to \infty} \frac{1}{n}\log |\mathcal{A}||\mathcal{B}| \nonumber 
\\
&\le \limsup_{n\to \infty} \frac{1}{n} \log \Big((n+1)^{|\mathcal{X}_1||\mathcal{X}_2|} |J^n(P^*_{X_1,X_2})|\Big)\label{eq:5.1}\\
&=\limsup_{n\to \infty} \frac{1}{n} \log |J^n(P^*_{X_1,X_2})|\label{eq:5.2}\\
&\le 
\limsup_{n\to \infty} \frac{1}{n} \log |\mathcal{A}^*||\mathcal{B}^*|\label{eq:5.3}
\\
&
\le 
\limsup_{n\to \infty} \frac{1}{n}\Big[\log\Big( 
\alpha\big(G^{nP^*_{X_1}(1)}_{1}\boxtimes G^{nP^*_{X_1}(2)}_{2}\boxtimes \cdots \boxtimes G^{nP^*_{X_1}(|\mathcal{X}_1|)}_{|\mathcal{X}_1|} \big)
\Big) 
+
\log\Big( \alpha\big(H^n[P^*_{X_1}]\big)\Big)
\Big]
\label{eq:6:3}
\\
&
\le 
\limsup_{n\to \infty} 
\min_{\eta, \eta' \in \Delta(\mathcal{G})}
\frac{1}{n}\Big[\log\Big( 
\eta\big(G^{nP^*_{X_1}(1)}_{1}\boxtimes G^{nP^*_{X_1}(2)}_{2}\boxtimes \cdots \boxtimes G^{nP^*_{X_1}(|\mathcal{X}_1|)}_{|\mathcal{X}_1|} \big)
\Big) 
+
\log\Big( \eta'\big(H^n[P^*_{X_1}]\big)\Big)
\Big]
\label{eq:6:4}
\\
&
\le 
\limsup_{n\to \infty} 
\min_{\eta, \eta' \in \Delta(\mathcal{G})}
\frac{1}{n}\Big[
\sum_{x_1\in \mathcal{X}_1} nP^*_{X_1}(x_1)\log\Big( 
\eta\big(G_{x_1}\big)
\Big) 
+
\log\Big( \eta'\big(H^n[P^*_{X_1}]\big)\Big)
\Big]
\label{eq:6:5}
\\
&\le
\max_{P_{X_1}} 
\sum_{x_1\in \mathcal{X}_1} P_{X_1}(x_1) \Theta(G_{x_1}) 
+
\Theta(H,P_{X_1})\label{eq:6:6}
\end{align}
where (\ref{eq:5.1}) follows from (\ref{eq-exist-J});
(\ref{eq:5.2}) follows from the fact that $|\mathcal{X}_1|$, $|\mathcal{X}_2|$ are fixed when $n$ tends to infinity;
(\ref{eq:5.3}) follows from (\ref{eq-J-leq});
\eqref{eq:6:3} follows from \eqref{eq-A_star} and \eqref{eq-B_star};
(\ref{eq:6:4}) follows from Theorem \ref{thm-asym-spec} that $\alpha(G)\le \min_{\eta\in \Delta(\mathcal{G})} \eta(G)$ for any graph $G$; 
(\ref{eq:6:5}) follows from Theorem \ref{thm-asym-spec} that each $\eta\in \Delta(\mathcal{G})$ is multiplicative with respect to the strong product; 
and (\ref{eq:6:6}) follows from Theorem \ref{thm-asym-spec} and Lemma~\ref{lemma-Vrana}. 

This completes the proof.
\end{proof}

In particular, considering the DM-TWC such that $|\mathcal{X}_1|=2$, $H=\overline{K}_2$, $G_1=\overline{K}_{|\mathcal{X}_2|}$ and $G_2=G$ is a general graph, we have the following result.
 
\begin{theorem}\label{thm-KG-KK}
$\mathscr{C}_{ze}([\overline{K}_{|\mathcal{X}_2|},G;\overline{K}_2,\ldots,\overline{K}_2])$ is contained in the region
\begin{equation*}
\begin{split}
\Big\{(R_1,R_2): R_1\ge 0, R_2\ge 0,
R_1+R_2\le \log \Big(|\mathcal{X}_2|+ 2^{\Theta(G)}\Big) \Big\}.
\end{split}
\end{equation*}
\end{theorem}

\begin{proof}
Recall that 
$\Theta(\overline{K}_n)=\log(n)$ and $\Theta(\overline{K}_n, P) = H(X)$. According to Theorem~\ref{thm-G_K}, we have 
\begin{align*}
R_1+R_2
&\le 
\max_{P_{X_1}} 
\sum_{x_1\in \mathcal{X}_1} P_{X_1}(x_1) \Theta(G_{x_1}) 
+
\Theta(\overline{K}_2,P_{X_1})
\\
&=
\max_{P_{X_1}} \big\{
P_{X_1}(1)\cdot \log|\mathcal{X}_2| + P_{X_1}(2)\cdot \Theta(G) + H(X_1) \big\}\\
&=\log \Big(|\mathcal{X}_2|+ 2^{\Theta(G)}\Big),
\end{align*}
where the last equality is achieved by taking 
$
P_{X_1}(1)=\frac{|\mathcal{X}_2|}{|\mathcal{X}_2|+ 2^{\Theta(G)}}$ and $
P_{X_1}(2)=\frac{2^{\Theta(G)}}{|\mathcal{X}_2|+ 2^{\Theta(G)}}$. 
\end{proof}

We remark that Theorem \ref{thm-KG-KK} (hence also Theorem \ref{thm-G_K}) could outperform Theorem \ref{upperbd-main}, see the following example.

\begin{example}\label{example-thm3-outperform-thm2}\rm
Consider the channel $[\overline{K}_{5},C_5;\overline{K}_2,\ldots,\overline{K}_2]$ where $C_5$ is the Pentagon graph.
It is well known from \cite{Lovasz1979,Shannon1956} that $\Theta(C_5)=\frac{1}{2}\log 5$.
Then by Theorem \ref{thm-KG-KK} we have an upper bound on the sum-rate $R_1+R_2\le \log (5+\sqrt{5}) \approx 2.8552$, while Theorem \ref{upperbd-main} only gives an upper bound $R_1+R_2\le 2.9069$.
\end{example}

\subsection{Explicit constructions}
\label{subsec-construction-certain-TWC}

In this subsection we present explicit constructions of uniquely decodable codebook pairs which could attain the outer bound of Theorem \ref{thm-KG-KK} in certain special cases.
\begin{theorem}\label{thm-construction-KG-KK}
Let $m$ be a prime power, $|\mathcal{X}_2|=q=ms$ and $G=K_m\sqcup\cdots\sqcup K_m$ be a disjoint union of $s$ cliques. 
Then 
$\mathscr{C}^{sum}_{ze}\Big([\overline{K}_q,G;\overline{K}_2,\ldots,\overline{K}_2]\Big)=\log (q+s)$.
\end{theorem}

\begin{proof}
First by Theorem \ref{thm-KG-KK}, we have an upper bound on the sum-capacity as
\begin{equation}\label{eq-upper-q+s}
\mathscr{C}^{sum}_{ze}\Big([\overline{K}_q,G;\overline{K}_2,\ldots,\overline{K}_2]\Big)\le \log \Big(|\mathcal{X}_2|+2^{\Theta(G)}\Big)=\log (q+s).
\end{equation}

Next we consider the lower bound. 
Notice that $G=K_m\boxtimes \overline{K}_s$. Accordingly we could reformulate the channel as 
\begin{equation*}
    [\overline{K}_q,G;\overline{K}_2,\ldots,\overline{K}_2]=[\overline{K}_m,K_m;\overline{K}_2,\ldots,\overline{K}_2]\boxtimes [\overline{K}_s;K_1,\ldots,K_1],
\end{equation*}
where the first $[\overline{K}_m,K_m;\overline{K}_2,\ldots,\overline{K}_2]$ corresponds to a channel with input alphabets $\mathcal{X}^{(1)}_1=\{1,2\}$ and $\mathcal{X}^{(1)}_2=\{1,\ldots,m\}$; and the second $[\overline{K}_s;K_1,\ldots,K_1]$ is with input alphabets $\mathcal{X}^{(2)}_1=\{1\}$ and $\mathcal{X}^{(2)}_2=\{1,\ldots,s\}$. 
Together with Lemma~\ref{lemma-product-sumcapacity}, we have 
\begin{equation}\label{eq-product-channel-ge}
    \mathscr{C}^{sum}_{ze}\Big([\overline{K}_q,G;\overline{K}_2,\ldots,\overline{K}_2]\Big)
    \geq \mathscr{C}^{sum}_{ze}\Big([\overline{K}_m,K_m;\overline{K}_2,\ldots,\overline{K}_2]\Big)+ \mathscr{C}^{sum}_{ze}\Big([\overline{K}_s;K_1,\ldots,K_1]\Big).
\end{equation}

On the one hand, it is easy to see that
\begin{equation}\label{eq-first-subchannel}
    \mathscr{C}^{sum}_{ze}\Big([\overline{K}_s;K_1,\ldots,K_1]\Big)=\log s
\end{equation} since this is a clean channel and Alice and Bob could always communicate without error. 
On the other hand, by Lemma~\ref{randomlinear}, we could obtain
\begin{equation}\label{eq-second-subchannel}
    \mathscr{C}^{sum}_{ze}\Big([\overline{K}_m,K_m;\overline{K}_2,\ldots,\overline{K}_2]\Big)\ge \log (m+1).
\end{equation}
In fact,
letting $q=m$, $q'=2$ and $\tau=1$ in Lemma \ref{randomlinear}, we have a pair $(\mathcal{C}, \Upsilon(\mathcal{C}))$ satisfying that
$\mathcal{C}$ is an $m$-ary $[n,k]$ linear code and  $\Upsilon(\mathcal{C})\subseteq \mathbb{F}^n_2$ such that
\begin{equation}\label{size-one}
    |\Upsilon(\mathcal{C})|\ge \binom{n}{k} \prod_{i=1}^{\infty} (1-m^{-i}).
\end{equation}
Now let $\mathcal{A}=\Upsilon(\mathcal{C})$ and $\mathcal{B}=\mathcal{C}$, then it is easy to see that $(\mathcal{A},\mathcal{B})$ is a uniquely decodable codebook pair with respect to the channel $[\overline{K}_m,K_m;\overline{K}_2,\ldots,\overline{K}_2]$. And the corresponding sum-rate is
\begin{equation*}
\begin{split}
\lim_{n\to \infty} \frac{1}{n}\log |\mathcal{A}||\mathcal{B}|&= \lim_{n\to \infty} \frac{1}{n}\bigg(\log m^k+ \log  \binom{n}{k} \prod_{i=1}^{\infty} (1-m^{-i})\bigg)\\
&=\lim_{n\to \infty} \frac{k}{n}\log m+ h\Big(\frac{k}{n}\Big).
\end{split}
\end{equation*}
Taking $k/n=m/(m+1)$, we obtain a lower bound $\log (m+1)$ on the best possible sum-rate, i.e.~\eqref{eq-second-subchannel}. 

Combining~\eqref{eq-product-channel-ge}-\eqref{eq-second-subchannel}, we have $\mathscr{C}^{sum}_{ze}\Big([\overline{K}_q,G;\overline{K}_2,\ldots,\overline{K}_2]\Big)\ge \log(m+1) + \log s = \log(q+s)$, which also implies an explicit uniquely decodable codebook pair for the channel $[\overline{K}_q,G;\overline{K}_2,\ldots,\overline{K}_2]$ based on the argument of Lemma~\ref{lemma-product-sumcapacity}. 
Then together with~\eqref{eq-upper-q+s} we complete the proof.  
\end{proof}

\section{Concluding remarks}
\label{sec-conclusion}

In this paper, we investigated the non-adaptive zero-error capacity region of the DM-TWC and provided several single-letter inner and outer bounds, some of which coincide in certain special cases. Determining the exact zero-error capacity region of a general DM-TWC remains an open problem, and a clearly difficult one, since it includes the notorious Shannon capacity of a graph as a special case. Despite this inherent difficulty, the problem is richer than graph capacity, and we belive it deserves further study in order to obtain tighter bounds and smarter constructions.

One appealing direction is to extend the Lov\'asz's semi-definite relaxation approach in order to obtain tighter outer bounds, mimicking the graph capacity case. This however does not seem to be a simple task. In particular, one may ask whether the natural  quantity $\rho(\{G_i\},\{H_j\})$ defined in \eqref{eq-sdp}, which upper-bounds the one-shot zero-error sum-capacity, is sub-multiplicative with respect to the graph strong product, in which case it would also serve as an upper bound for the zero-error sum-capacity. This is however not evident, in part since the problem \eqref{eq-sdp} is not a semi-definite program. We have also considered other variations of the program \eqref{eq-sdp}. In particular, we have attempted to
modify the non-linear constraints $\langle E_{i,j},\Gamma \rangle \langle E_{i,m},\Gamma\rangle=0$ to be of a linear form
$\langle A,\Gamma \rangle=0$ for some suitable symmetric matrix $A$.
We have also looked at some variants of the orthonormal representation. For example, we considered the case where each graph vertex $a$ is labelled by a unit vector $v_a$, and if two vertices $a$ and $a'$ are nonadjacent $a\overset{b}{\nsim}a'$ if and only if $b\in F$ for some set $F$, then the vector projections of $v_a$ and $v_{a'}$ onto the subspace spanned by the vectors in $F$ are orthogonal. However, proving sub-multiplicativity in any of these settings has so far resisted are best efforts.

It would be also of much interest to consider the adaptive zero-error capacity of DM-TWC. Allowing Alice and Bob to adapt their transmissions on the fly can in general enlarge the zero-error capacity region. As a simple example, note that a point-to-point channel with noiseless feedback is a special case of the DM-TWC (where Bob has no information to send). In~\cite{Shannon1956}, Shannon explicitly derived the zero-error capacity with feedback for the point-to-point channel, and pointed out that for the channel corresponding to Pentagon graph this capacity is given by $\log(5/2)\approx 1.32$. This is strictly larger than the zero-error capacity without feedback $(\log5)/2\approx 1.16$, which can be thought of in this case as the non-adaptive zero-error capacity of the channel. Exploring the differences between the adaptive and non-adaptive zero-error capacity regions of a general DM-TWC remains a challenging future work.

\section*{Acknowledgments}

We would like to thank Sihuang Hu and Lele Wang for some helpful discussions on the generalization of Lov\'asz theta number.

\end{document}